\documentclass[english,12pt,reqno]{smfart}
\usepackage{amsfonts}
\usepackage{amsmath}
\usepackage{amssymb}
\usepackage{amscd}
\usepackage{latexsym}
\usepackage{graphicx}
\usepackage{hyperref}

\numberwithin{equation}{section}

\newcount\hh
\newcount\mm
\mm=\time
\hh=\time
\divide\hh by 60
\divide\mm by 60
\multiply\mm by 60
\mm=-\mm
\advance\mm by \time
\def\thistime{\number\hh:\ifnum\mm<10{}0\fi\number\mm}
\newtheorem{thm}[equation]{Theorem}

\newtheorem{lem}[equation]{Lemma}

\newtheorem{exs}[equation]{Examples}
\newtheorem{ex}[equation]{Example}
\newtheorem{remark}[equation]{Remark}

\numberwithin{equation}{section}

\def\nn{\nonumber}

\def\Li#1(#2){\textrm{Li}_{#1}\left(#2\right)}
\def\cLi_#1(#2){\mathcal{L}_{#1}\left(#2\right)}
\def\bLi_#1(#2){\mathbf{L}_{#1}\left(#2\right)}

\def\suns{\circleddash}
\def\As{A_\circleddash}

\def\cIs{\mathcal{I}_\circleddash}

\def\cJs{\mathcal{J}_\circleddash}

\def\ZZ{{\mathbb Z}}
\def\IC{{\mathbb C}}
\def\IR{{\mathbb R}}

\def\cD{\mathcal{D}}

\def\cE{\mathcal{E}}

\def\Imm{\Im\textrm{m}}
\def\cEs{\mathcal E_\circleddash}
\newcommand{\eq}[2]{\begin{equation}\label{#1}#2 \end{equation}}
\newcommand{\ml}[2]{\begin{multline}\label{#1}#2 \end{multline}}
\newcommand{\ga}[2]{\begin{gather}\label{#1}#2 \end{gather}}

\newcommand{\surj}{\twoheadrightarrow}
\newcommand{\inj}{\hookrightarrow}

\newcommand{\sD}{{\mathcal D}}
\newcommand{\sE}{{\mathcal E}}

\newcommand{\sO}{{\mathcal O}}

\newcommand{\sV}{{\mathcal V}}
\newcommand{\sW}{{\mathcal W}}

\newcommand{\C}{{\mathbb C}}

\newcommand{\E}{{\mathbb E}}
\newcommand{\F}{{\mathbb F}}

\renewcommand{\H}{{\mathbb H}}

\renewcommand{\P}{{\mathbb P}}
\newcommand{\Q}{{\mathbb Q}}
\newcommand{\R}{{\mathbb R}}

\newcommand{\V}{{\mathbb V}}
\newcommand{\W}{{\mathbb W}}

\newcommand{\Z}{{\mathbb Z}}

\newcommand{\ve}{{\varepsilon}}
\newcommand{\nnn}{\newline\newline\noindent}

\hypersetup{colorlinks=true,backref=true,linkcolor=black,anchorcolor=black,citecolor=black,filecolor=black,menucolor=black,pagecolor=black,urlcolor=black}

\title[The sunset graph]{\bf The elliptic dilogarithm  for the sunset graph}

\author{Spencer Bloch}
\address{5765 S. Blackstone Ave., Chicago, IL 60637, USA}
\email{spencer\_bloch@yahoo.com}

\author{Pierre Vanhove}
 \address{
 Institut des Hautes Etudes Scientifiques\\
 Le Bois-Marie, 35 route de Chartres\\
 F-91440 Bures-sur-Yvette, France\hfill\break
Institut de Physique Th{\'e}orique\\
CEA, IPhT, F-91191 Gif-sur-Yvette, France\\
CNRS, URA 2306, F-91191 Gif-sur-Yvette, France}
\email{pierre.vanhove@cea.fr}

\thanks{IPHT-T/13/217, IHES/P/13/24}
\date{\today}

\begin{document}

 \begin{abstract}
We study the sunset graph defined as the scalar two-point self-energy at
two-loop order. We evaluated the sunset integral for all identical internal
masses in two dimensions. 
We give two calculations for the sunset amplitude; one based on an
interpretation of the amplitude as an inhomogeneous solution of a
classical Picard-Fuchs differential equation, and the other using
arithmetic algebraic geometry, motivic cohomology, and Eisenstein
series. Both methods use the rather special fact that the amplitude in
this case is a family of periods associated to the universal family of
elliptic curves over the modular curve $X_1(6)$. 
 We show that the integral is given by an elliptic dilogarithm
 evaluated at a sixth root of unity modulo periods. We
 explain as well how this elliptic dilogarithm value is related to the
 regulator of a class in the motivic cohomology of the universal
 elliptic family. 
\end{abstract}
\maketitle
\tableofcontents

%%%%%%%%%%%%%%%%%%%%%%%%%%%%%%%%%%%%%%%%%%%%%%%%%%%%%%%%%%%%%%%%%
%--------------------------------------------------------------------------
\section{Introduction}

Scattering amplitudes are fundamental objects describing how particles
interact. At a given loop order in the perturbative expansion in the
coupling constant, there are many ways of constructing the amplitudes
from first principles of quantum field theory. The result is an
algebraic integral with parameters, and the physical problem of
efficient evaluation of the integral is linked to the qualitative
mathematical problem of classifying these multi-valued functions of
the complexified kinematic invariants. The amplitudes are locally
analytic, presenting branch points at the thresholds where particles
can appear.

These questions can be studied order by order in perturbation. At
one-loop order, around a four dimensional space-time, all the
scattering amplitudes can be expanded in a basis of integral functions
given by box, triangle, bubbles and tadpole integral functions,
together with rational term contributions~\cite{Bern:1997sc,Britto:2004nc,Ossola:2006us} (see~\cite{Bern:1996je,Britto:2010xq} for some reviews on this subject).

The finite part of the $\epsilon=(4-D)/2$ expansion of the box and
triangle integral functions is given by dilogarithms of the proper
combination of kinematic invariants. The finite part of the bubble and
tadpole integral is a logarithm function of the physical parameters.

The appearance of the dilogarithm and logarithms at one-loop order is
predictable from unitarity considerations since this reproduces the
behaviour of the one-loop scattering amplitude under single, or
double  two-particle cuts in four dimensions.

The fact that one-loop amplitudes are   expressed as dilogarithms and
logarithms can as well be understood
motivitically~\cite{Bloch:2005bh,Bloch:2010gk}, but the status of
two-loop order scattering amplitude is far less well understood for
generic amplitudes (see for instance~\cite{Johansson:2013sda,Johansson:2012zv,CaronHuot:2012ab,Kosower:2011ty} for some recent progress).

The sunset integral arises as the two-loop self-energy diagram in the
evaluation of higher-order correction in QED, QCD or electroweak
theory precision calculations~\cite{Bauberger:1994nk}, or as a sub-topology of  higher-order
computation~\cite{CaronHuot:2012ab}. As a consequence, it has been the subject of numerous analyses. 
The integral for various configurations of vanishing masses has been analyzed using
  the Mellin-Barnes methods in~\cite{Smirnov}, with two different masses and
  three equal masses in~\cite{Kniehl:2005bc}.
  An asymptotic expansion of the sunset integral has been given
  in~\cite{Broadhurst:1993mw}.  Various forms for the integral
  have been considered either in geometrical
  terms~\cite{Davydychev:1997wa}, displaying some relations to one-loop
  amplitude~\cite{Davydychev:1995mq}, or a  
  representation in terms of hypergeometric function as given 
  in~\cite{Broadhurst:1987ei,Tarasov:2006nk,BorweinC,Kalmykov:2008ge} or as an
  integral of Bessel   functions as in~\cite{Bailey:2008ib,Broadhurst:2008mx}. 
Or a differential equation approach (in close relation to the method used in
section~\ref{sec:PF} of the present work) was
  considered
  in~\cite{Caffo:1998du,Laporta:2004rb,MullerStach:2011ru,Adams:2013nia,Groote:2012pa}.
We refer to these papers
   for a more complete list of references. 

The question of whether the sunset integral can expressed in terms
of known mathematical functions like polylogarithms
has not so far been addressed.

In order to answer  this  question will consider the sunset graph in
two space-time dimensions depicted in figure~\ref{fig:sunsetfig}. 
 The sunset integral with internal
masses in two dimensions is a completely finite integral free of
infra-red and ultra-violet divergences. Working with a finite integral
 will ease the discussion of the  mathematical nature of this
 integral.

 Although the ultimate goal is to understand the properties of
 two-loop amplitudes around four dimensional space-time, the
 restriction to two dimensions is not too bad since dimension shifting
 formulas, given in~\cite{Laporta:2004rb}, relate the result in two dimensions to the finite part of
 the integral in four dimensions.

Another restriction of this work is to focus only on the all equal masses case with all
internal masses  non zero and positive. We will find in this case that the sunset integral  is nicely expressed~\eqref{e:amplitude} in terms of elliptic
dilogarithms obtained by $q$-averaging values of the dilogarithm at a
sixth root of unity in~\eqref{e:Hsunset} with the following $q$-expansion
\begin{equation}
  E_\suns(q) =\sum_{n\in\mathbb Z\backslash\{0\}} \, {(-1)^{n-1}\over2\, n^2}\,
  \left(\sin\left({n\pi\over3}\right)+\sin\left({2n\pi\over3}\right)\right)\,{1\over 1-q^n} \,.
\end{equation}
This expression is locally analytic in $q$ and differs from the one for the regulator map \cite{BlochCMR} that would be expressed in terms of the (non-analytic) Bloch-Wigner dilogarithm  $D(z)$.  The reason for this
difference is explained in section~\ref{section:motiveHS} when evaluating
the amplitude using motivic methods (compare equations \eqref{31a} and \eqref{35}).

We give two calculations for the sunset amplitude; one based on an
interpretation of the amplitude as an inhomogeneous solution of a
classical Picard-Fuchs differential equation (in section~\ref{sec:PF}), and the other using
arithmetic algebraic geometry (in section~\ref{section:motiveHS}), motivic cohomology (in section~\ref{sec:motive}), and Eisenstein
series (in section~\ref{sec:eis}). Both methods use the rather special fact that the amplitude in
this equal mass case is a family of periods associated to the universal family of
elliptic curves over the modular curve $X_1(6)$. The elliptic fibres
$\sE_t$ are naturally embedded in $\P^2$ and pass through the vertices
$(1,0,0), (0,1,0), (0,0,1)$ of the coordinate triangle $\Delta:
XYZ=0$. Let $P \to \P^2$ be the blowup of these three points, so the
inverse image of $\Delta$ is a hexagon $\mathfrak h \subset P$. Then
$\sE_t$ lifts to $P$, and the amplitude period is closely related to
the cohomology 
\eq{}{H^2(P-\sE_t, \mathfrak h^0) 
}
where $\mathfrak h^0 := \mathfrak h - \mathfrak h\cap \sE_t$. 

The motivic picture which emerges from the basic sunset $D=2$ equal
mass case applies as well to all Feynman amplitudes. Quite generally,
the Feynman amplitude will be the solution to some sort of
inhomogeneous Picard-Fuchs equation $PF(x)=f(x)$ where $x$ represents
physical parameters like masses and kinematic invariants. The motive
determines the function $f$. In the equal mass sunset case, $f$ is
constant. In this case, the amplitude motive reduces to an extension of the
motive of an elliptic curve by a constant Tate motive. If the masses
are distinct, the elliptic curve motive is extended by a Kummer motive
which is itself an extension of constant Tate motives. The function
$f$ then involves a non-constant logarithm.  

In general, our idea is to relate the physical amplitude to a
regulator in the sense of arithmetic algebraic geometry applied to a
class in motivic cohomology. One has a hypersurface (depending on
kinematic parameters) ({\it cf.} equation~\eqref{e:Asunset} in the sunset
case) in projective space $\P^n$. One has an integrand ({\it cf.}
equation~\eqref{e:IsunsetSym2d} in the sunset case) which is a top degree meromorphic form on $\P^n$ with a pole along $X$, and one has a chain of integration \eqref{e:Ddef}. Let $\Delta:x_0x_1\cdots x_n=0$ be the coordinate simplex in $\P^n$. One first blows up faces of $\Delta$ on $\P^n$ in such a way that the strict transform $Y$ of $X$ meets all faces properly. In the sunset case, $Y=X$ and the blowup yields $P \to \P^2$, the blowup of the three vertices of the triangle $\Delta$. The total transform $\frak h$ of $\Delta$ in $P$ is a hexagon in the sunset case. Next, one constructs a motivic cochain which is an algebraic cycle $\Xi$ on $P\times (\P^1-\{0,\infty\})^{n-1}$ of dimension $n-1$. One then tries to interpret $\Xi$ as a relative motivic cohomology class in $H^{n+1}_M(P,Y,\Q(n))$. This is possible in the sunset case, but in general one has a closed $Z\subset Y$ and $\Xi$ represents a motivic cohomology class in $H^{n+1}_M(P-Z,Y-Z,\Q(n))$ (cf. \cite{MullerStach:2012mp}). When $Z=\emptyset$, the inhomogeneous term $f$ of the differential equation will be relatively simple, but the study of $Z$ more generally should involve shrinking edges of the graph (cf. op. cit.).
%-------------------------------------------------------------------------
\section{The sunset integral}
\label{sec:sunsetintegral}

\begin{figure}[ht]
  \centering
  \includegraphics[width=8cm]{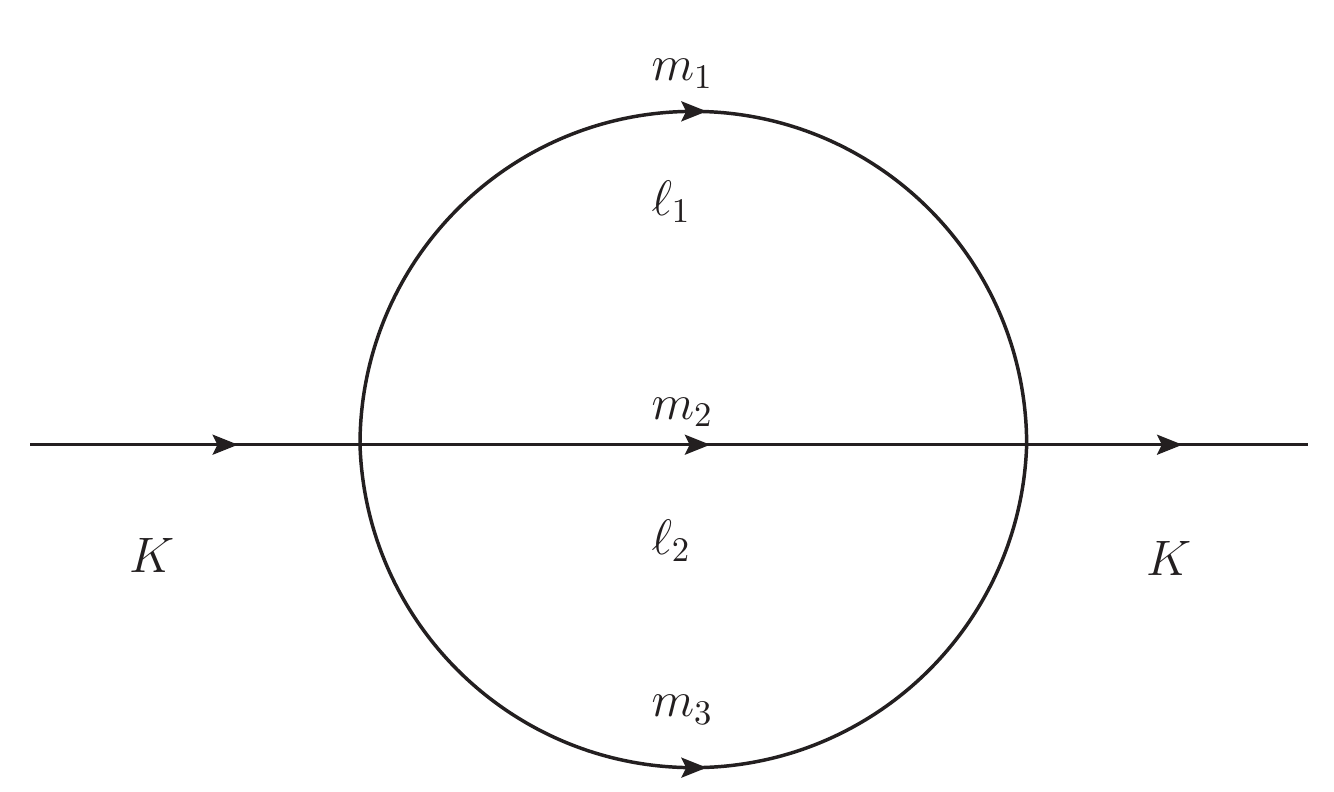}
  \caption{The two-loop sunset graph with three different masses. The
    external momentum is $K$, and the loop momenta are $\ell_1$ and $\ell_2$.}
  \label{fig:sunsetfig}
\end{figure}

The integral associated to the sunset graph given in figure~\ref{fig:sunsetfig}
\begin{equation}
  \label{e:GraphGeneralFeynman}
   \cIs^D(K^2,m_1^2,m_2^2,m_3^2)= (\mu^2)^{3-D} \int_{\IR^{2,2(D-1)}} {
      \delta(K=\ell_1+\ell_2+\ell_3)\over d_1^{2} d_2^{2} d_3^{2}}\,d^D\ell_1 d^D\ell_2\,,
\end{equation}
where  $d_i^2=\ell_i^2-m_i^2+i\varepsilon$ with $\varepsilon>0$  and $m_i\in \IR^+$ and
$\ell_i$ and $K$ are in $\IR^{1,D-1}$ with signature $(+
-\cdots-)$.  With this choice of
metric,   the physical region corresponds to $K^2\geq0$. In the following we will consider the Wick rotated
integral with $\ell_i\in \IR^D$.   We have
introduced the arbitrary scale $\mu^2$ of mass-dimension two so that
the integral is dimensionless.

The Feynman
parametrisation is easily  obtained following~\cite[sec.~6-2-3]{Itzykson:1980rh}
\begin{equation}
  \label{e:IsunsetAlt}
  \cIs^{D}= {\pi^D\Gamma(3-D)\over (\mu^2)^{D-3}}\, \int_0^\infty
  \int_0^\infty {(x+y+xy)^{3(2-D)\over2} \,dx dy\over ((m_1^2x+m_2^2y+m_3^2)(x+y+xy)  -K^2 xy)^{3-D}}\,.
\end{equation}
For geometric reasons, we will use
an alternative symmetric representation of the integral formulated
using homogeneous coordinates 
(see~\cite{MullerStach:2011ru,Adams:2013nia}) 
\begin{equation}
  \label{e:IsunsetSym}
  \cIs^{D}= {\pi^D\Gamma(3-D)\over (\mu^2)^{D-3}}\, \int_{\cD}
  {(xz+yz+xy)^{3(2-D)\over2} \,(x dy\wedge dz-y dx\wedge
    dz+z dx\wedge dy)\over ((m_1^2x+m_2^2y+m_3^2z)(xz+yz+xy)  -K^2 xyz)^{3-D}}\,.
\end{equation}
with the domain of integration
\begin{equation}\label{e:Ddef}
  \cD=\{(x,y,z)\in \mathbb P^2| x,y,z\geq0\}  \,.
\end{equation}

The physical parameters of the sunset integral enter the polynomial
\begin{equation}
  \label{e:Asunset}
\As(x,y,z,m_1,m_2,m_3,K^2):=( m_1^2 x+ m_2^2 y+m_3^2z)(xz+xy+yz)-K^2 xyz\,.
\end{equation}
For generic values of the parameters the curve $\As(x,y,z)=0$ defines an
elliptic curve
\begin{equation}
  \cEs(m_1,m_2,m_3,K^2):=\{\As(x,y,z,m_1,m_2,m_3,K^2)=0| (x,y,z)\in \mathbb P^2  \}\,.
\end{equation}
 When the physical parameters vary this defines a
family of elliptic curves.  For special values the elliptic curve
degenerates. 

\medskip
Before embarking on a general discussion of the property of
the elliptic curve,  we discuss in section~\ref{sec:special} the degenerate points $K^2=0$ and the
pseudo-thresholds $K^2\in\{(-m_1+m_2+m_3)^2, (m_1-m_2+m_3)^2, (m_1+m_2-m_3)^2\}$.

\smallskip
In this paper we will only  consider the all equal  masses case
$m_1^2=m_2^2=m_3^2=m^2$ and $K^2=t\, m^2$ and evaluate the integral in $D=2$ space-time
dimensions 
\begin{equation}
  \label{e:IsunsetSym2d}
  \cIs^{2}= {\pi^2 \mu^2\over m^2}\, \int_{\cD}
  {x dy\wedge dz-y dx\wedge
    dz+z dx\wedge dy\over (x+y+z)(xz+yz+xy)  -t xyz}\,.
\end{equation}
The restriction to two dimensions will make contact with the
motivic analysis easier since the sunset integral is both ultraviolet and infrared finite for
non-vanishing value of the internal mass.

\smallskip
If one is interested by the expression of the integral in 
 four space-time dimensions one needs perform an $\epsilon=(4-D)/2$
 expansion because the all equal masses sunset
integral is not finite due to ultraviolet divergences
\begin{equation}
  \cIs^{4-2\epsilon}(K^2,m^2)=16 \pi^{4-2\epsilon}
  \Gamma(1+\epsilon)^2 \,\left(m^2\over \mu^2\right)^{1-2\epsilon}\,\left(
  {a_2\over\epsilon^2}+{a_1\over\epsilon} +a_0+ O(\epsilon)\right)\,.
\end{equation}
 It was shown in~\cite{Laporta:2004rb} that the
 coefficients of this $\epsilon$-expansion are given by 
derivatives of the  two dimensional sunset integral, $\cJs^2(t)=m^2/(\pi^2\mu^2)\,\cIs^2$,
with the first coefficients given by 
\begin{eqnarray}
  a_2&=&-{3\over8}\cr
a_1&=&{18-t\over 32}\cr
a_0&=& {(t-1)(t-9)\over12} (1+ (t+3) {d\over dt}) \cJs^2(t)+
{13t-72\over 128}\,.
\end{eqnarray}
%

%----------------------------------------------------------------------------------
\section{Special values}\label{sec:special}
In this section, we consider the amplitude for some special values of $K^2$. 
%----------------------------------------------------------------------------------
\subsection{The case $K^2=0$ with three non vanishing masses}
 \label{sec:K2zero}

We evaluate the integral at some
special values $K^2=0$, assuming that the internal
masses $m_i$ are non-vanishing and positive.  
In $D=2$ dimensions the integral is finite and reads
 \begin{equation}\label{e:I0m1m2m3}
  \cIs^2(0,m_i)=\pi^2\mu^2 \int_{0}^{+\infty}\int_{0}^{+\infty}
 \,{  dxdy \over
    (x+y+xy)(m_1^2x+m_2^2 y+ m_3^2)}\,.
 \end{equation}
 In this form the invariance under the permutations of the masses is
 seen as follows.  The exchange between $m_1$ and $m_2$  corresponds to
 the transformation $(x,y)\to(y,x)$ and the exchange between
 $m_1$ and $m_3$ corresponds to the transformation $(x,y)\to (1/x,y/x)$.

Performing the integration over $y$ we obtain
 \begin{equation}\label{e:I2K2zero}
   \cIs^2(0,m_i^2)={\pi^2\mu^2\over   m_1^2}\, J_0\left({m_1^2\over
         m_2^2},{m_3^2\over
         m_2^2}\right)\,,
 \end{equation}
 with
\begin{eqnarray}
  J_0(a,b)&=&a \, \int_0^\infty {\log R(a,b,x)\over x
    R(a,b,x)-x}\, dx\cr
R(a,b,x)&=&(1+x^{-1})(ax+b) \,,
\end{eqnarray}
We introduce $z$ and $\bar z$ such that 
\begin{equation}
  x R(a,b,x)-x= (1+x)(ax+b)-x= a(x+z)(x+\bar z)\,,  
\end{equation}
with $(1-z)(1-\bar z)= 1/a=m_2^2/m_1^2$ and $z\bar z
=b/a=m_3^2/m_1^2$.
The roots  $z$ and $\bar z$ are given by 
\begin{eqnarray}\label{e:zdef}
  z&:=& {m_1^2-m_2^2+m_3^2+\sqrt\Delta\over 2m_1^2} \\
\bar z&:=&{m_1^2-m_2^2+m_3^2-\sqrt\Delta\over2m_1^2}\,.
\end{eqnarray}
where the discriminant 
\begin{equation}\label{e:Deltadef}
\Delta :=\sum_{i=1}^3 m_i^4- 2\sum_{1\leq i<j\leq
  3}m_i^2m_j^2\,.
\end{equation}

The discriminant vanishes $\Delta=0$ when $m_3=m_1+m_2$ or any
permutation of the masses (we assume all the masses to be positive).

%----------------------------------------------------------------------------------
\subsubsection{The case $\Delta\neq0$}
\label{sec:deltaneq0}
In this case in terms of these variables the integral becomes
\begin{eqnarray}
  J_0(z)&=&\int_0^\infty \, j_0(z,\bar z,x)\, dx\\  
\nn j_0(z,\bar z,x)&=& {-\log x+\log(1+x)-\log(1-z)-\log(1-\bar
  z)+\log(x+z\bar z)\over (x+z)(x+\bar z)}\,.
\end{eqnarray}
This integral is easily evaluated, and found to be given by 
\begin{equation}
  \label{e:J0result}
  J_0(z)= {4iD(z)\over z-\bar z}\,,
\end{equation}
where $D(z)$ is
the Bloch-Wigner dilogarithm function  (we refer to
appendix~\ref{sec:ellipticdilog} for some details about this function).
Therefore the sunset integral $\cIs^2(0,m_i^2)$ is given by the
Bloch-Wigner dilogarithm
\begin{equation}\label{e:J0}
    \cIs^2(0,m_i^2)= 2 i\pi^2\,\mu^2\,   {D(z)\over \sqrt{\Delta}}\,.
\end{equation}

Under the permutation of the three masses $(m_1,m_2,m_3)$ the variable
$z$ ranges over the orbit
\begin{equation}
\Big\{z, 1-\bar z, {1\over\bar z},1-{1\over z}, {1\over 1-z}, -{\bar
  z\over 1-\bar z}\Big\}\,.
\end{equation}
Since $\Delta$ is invariant under the permutation of the masses and,
thanks to 
the six functional equations for the Bloch-Wigner function
in~\eqref{e:Dfunc}, the amplitude in~\eqref{e:J0} is invariant.

%----------------------------------------------------------------------------------
\subsubsection{The case $\Delta=0$}
\label{sec:case-delta=0}

Suppose for example  $m_3=m_1+m_2$ then the roots in
eq.~\eqref{e:zdef} become  $z=\bar z=(m_1+m_2)/m_1$.  The integral is
easily evaluated and yields 
\begin{multline}
       \cIs^2(0,m_1^2,m_2^2,(m_1+m_2)^2)=  4\pi^2\mu^2\, \Big(
       {\log(m_1+m_2)\over m_1m_2}-{\log(m_1)\over m_2(m_1+m_2)}\cr
-{\log(m_2)\over m_1(m_1+m_2)}\Big)\,.
\end{multline}
The other cases are obtained with a permutation of the masses.

%----------------------------------------------------------------------------------
\subsection{At the pseudo-thresholds with three internal masses}
\label{sec:pseudothresholds}

We evaluate the integral at the pseudo-thresholds
$K^2=(m_1+m_2+m_3-2m_i)^2$ with $i=1,2,3$.
Since  the sunset integral is invariant under the permutation of the
internal masses we can assume that $0<m_1\leq m_2\leq m_3$.

By combining direct integrations and numerical evaluations of the
integrals for various values of the internal masses we find the
expressions for the sunset integral at the pseudo-thresholds.

If we introduce the notation $M_i= m_1+m_2+m_3-2m_i\neq0$ (the case
$M_i=0$ has been evaluated in section~\ref{sec:case-delta=0}), we then have
for $i=1,2,3$ 
\begin{equation}
  \label{e:Fiberppm}
  \cIs^2(m_1^2,m_2^2,m_3^2,M_i^2)= {- 2i \pi^2\mu^2\over
    \sqrt{m_1m_2m_3 M_i}} \times  \sum_{j=1}^3 {\partial M_i\over\partial m_j}\times\tilde D\left(\sqrt{m_j^2
    M_i\over m_1m_2m_3}\right)\,,
\end{equation}
where we have defined $\tilde D(z)$ as 
\begin{equation}
 \tilde D(z)=
 {1\over2i}\left(\Li_2(z)-\Li_2(-z)+\frac12\, \log(z^2)\,
  \log\left(1-z\over 1+z\right)\right)\,.
\end{equation}
%
%----------------------------------------------------------------------------------
\subsection{The all equal masses case}

When all the internal masses are identical $m_1=m_2=m_3=m$ we set
$t=K^2/m^2$.    We evaluate the integral at the special
values $t=0$ and $t=1$ for further reference in the complete
evaluation done in section~\ref{sec:PF}.
%----------------------------------------------------------------------------------
\subsubsection{$t=0$ case}
\label{sec:t=0-case}

For $t=0$  the integral has been evaluated in
section~\ref{sec:K2zero} for generic values of the internal masses.
When all the internal masses are equal $m_1=m_2=m_3=m$,  then
the root in eq.~\eqref{e:zdef}  become the 
sixth-root of unity $z=\zeta_6:=(1+i\sqrt{3})/2$.  The sunset integral in~\eqref{e:J0} then
evaluates to\footnote{ Since $D(\zeta_6)= Cl_2(\pi/3)$, the second
  Clausen sum,  this evaluation is in agreement
with the result of~\cite{Laporta:2004rb}.  }
\begin{equation}\label{e:t=0}
    \cIs^2(0,m^2)=i \pi^2{ \mu^2\over m^2}\, {D(\zeta_6)\over \Imm(\zeta_6)}\,.
\end{equation}
Using eq.~\eqref{e:Drel} one can rewrite this expression for the
integral as
\begin{equation}
   \cIs^2(0,m^2)={6i\pi^2\over\sqrt3}{ \mu^2\over m^2}\,  \sum_{n\geq1} \,\left(\sin\left({n\pi\over3}\right)-\sin\left({2n\pi\over3}\right)\right)
   {1\over n^2}\,.
\end{equation}
or in a form that will be useful later
\begin{equation}\label{e:I0li2}
   \cIs^2(0,m^2)=-{8\pi^2\over5}{ \mu^2\over m^2}\, \left(\Li2(\zeta_6^5)+\Li2(\zeta_6^4)-\Li2(\zeta_6^2)-\Li2(\zeta_6)\right)\,.
\end{equation}

%----------------------------------------------------------------------------------
\subsubsection{$t=1$ case}
\label{sec:t=1-case}

The case  $t=K^2/m^2=1$ corresponds to the one-particle pseudo-threshold when one
internal line is cut. At this value the integral takes the form
\begin{equation}
  \cIs^{2}(K^2=m^2,m^2)= \pi^2 {\mu^2\over m^2}\, \int_0^\infty\int_0^\infty
  {dxdy\over (x+1)(y+1)(x+y)}\,.
\end{equation}
This integral is readily evaluated to give
\begin{equation}
  \label{e:t=1}
   \cIs^2(m^2,m^2)=2\pi^2{\mu^2\over m^2}\,{\pi^2\over8}\,.
\end{equation}
One checks that this expression is a special case of
section~\ref{sec:pseudothresholds} since $\tilde D(1)=-i\pi^2/8$.
%--------------------------------------------------------------------------
\section{Families of elliptic curves}\label{sec:K3fam}

We now turn to the discussion of the nature of the sunset integral for
generic values of $t=K^2/m^2$.
We define the integral
\begin{equation}\label{e:Js}
  \cJs^2(t):= {m^2\over\mu^2\pi^2} \, \cIs^2(K^2,m^2) = \int_0^\infty\int_0^\infty  {dx
    dy\over \As(x,y,t)}\,,
\end{equation}
and 
\begin{equation}
  \label{e:Asunsett}
\As(x,y,t)=(1+x+y)(x+y+xy)-xy t\,.
\end{equation}
We have a  family of elliptic curves
\begin{equation}\label{e:Es}
\cE_t:=\{\As(x,y,t)=0, x,y\in\mathbb P^2\}
\end{equation}
for $t\in\mathbb P^1$.
The  following change of variables 
\begin{eqnarray}
  \label{e:iso}
   (x,y,z)&=&\left(\eta\, {2\zeta+\sigma(t-1)\over
     2\sigma(\eta-\sigma)},\eta\,{-2\zeta+\sigma(t-1)\over
     2\sigma(\eta-\sigma)},1\right)\\
\nn(\sigma,\zeta,\eta)&=&\left({x+y+z(1-t)\over x+y},{(t-1)(x-y)(x+y+z(1-t))\over 2(x+y)^2},1\right)\,,
 \end{eqnarray}
brings  the elliptic curve into its standard  Weierstra\ss{} form 
\begin{equation}\label{e:WeierstrassSunset}
  \zeta^2\eta=\sigma(t \eta^2 + {(t-3)^2-12\over4} \,\sigma\eta +\sigma^2)\,.
\end{equation}
The discriminant is given by 
\begin{equation}
  \label{e:Disc}
  \Delta=256\, (t-9) (t-1)^3 t^2\,,
\end{equation}
and the  $j$-invariant is given by 
\begin{equation}
   J(t)= \frac{(t -3)^3 (t^3-9t^2+3t-3)^3}{(t -9) (t -1)^3 t ^2}\,. 
\end{equation}

This family of elliptic curves 
\begin{equation}
  \mathcal S_\circleddash=\{(x,y,z,t)\in \mathbb P^2\times \mathbb P^1,  \As(x,y,z,t)=0\}  \,,
\end{equation}
defines a pencil of elliptic curves in $\mathbb P^2$ corresponding to a modular family of
elliptic curves $f: \sE \to X_1(6)=\{ \tau\in\mathbb C
|\Imm(\tau)>0\}/\Gamma_1(6)$.  This will play an important 
role in section~\ref{sec:eis} when expressing the motive for the sunset graph in term of the
Eisenstein series in $E_\suns(q)$. 

 From the discriminant in eq.~\eqref{e:Disc} we deduce
there are  four singular fibers at $t=0,1,9,\infty$ of
respective 
Kodaira type  $I_2,I_3,I_1,I_6$ listed by Beauville
in~\cite{beauville}. The total space is a rational elliptic surface as shown in~\cite{Beukers}. 

\medskip

For a given $t\in \mathbb P^1$, the elliptic curve $\cE_t$ intersects the domain of integration
$\cD$ in eq.~\eqref{e:Ddef} at  the following six points
\begin{equation}
  \label{e:6points}
  (x,y,z)=\Big\{(0,1,0), (1,0,0), (0,0,1),
(1,-1,0), (0,-1,1), (-1,0,1)\Big\}
\end{equation}
which are mapped to using~\eqref{e:iso}
\begin{multline}
  \label{e:6points2}
  (\sigma,\zeta,\eta)=\Big\{(0,1,0),(0,0,1),(1,{t-1\over2},1),(1,{1-t\over2},1),\cr
(t,{t(t-1)\over2},1),(t,{t(1-t)\over2},1)\Big\}\,.
\end{multline}
\begin{lem}\label{lem:torsion}
The points in the list in~\eqref{e:6points2} are always of torsion for
any values of $t$ of the elliptic curve in~\eqref{e:WeierstrassSunset}.
 The torsion group   is $\ZZ/6\ZZ$ generated by $(t,t(t-1)/2,1)$ or $(t,-t(t-1)/2,1)$.
\end{lem}
\begin{proof}
As is standard for the Weierstra\ss{} model, we set the point $O=(0,1,0)$ to be the identity.
Then $(0,0,1)$ is of order 2 (any points of the form $(x,0,1)$ are of
order 2).
 A point $P=(x,y,1)$ is of order 3, if and only if $3P=O$ or
 equivalently $2P=-P$, which  implies for the elliptic curve in eq.~\eqref{e:WeierstrassSunset} that 
 \begin{equation}
   x=    \frac{\left(t-x^2\right)^2}{x \left(((t-6) t-3) x+4 t+4 x^2\right)}
 \end{equation}
This has a solution $x=1$ and $y=\pm(1-t)/2$, so the points
$P_3=(1,{t-1\over2},1)$ and $P_4=(1,{1-t\over2},1)$ are of order 3. For the remaining two
points $P_5=(t,{t(t-1)\over2},1)$ and $P_6=(t,{t(1-t)\over2},1)$, one
checks that $2P_5=P_3$ and $2P_6=P_4$. Therefore the points $P_5$ and
$P_6$ are of order 6.  We have as well that $5P_5=
2P_3+P_5=P_5-P_3=P_6$ and $5P_6=P_5$. We conclude that the torsion
group is $\ZZ/6\ZZ$ generated by $P_5$ or $P_6$.
\end{proof}

\medskip
The Picard-Fuchs equation associated with this family of elliptic
curves given in~\cite{Beukers} reads for $t\in \mathbb P^1 \backslash \{0,1,9,+\infty\}$
\begin{equation}\label{e:PFt2}
  L_t F(t):= {d\over dt} \left( t(t-1)(t-9)  {dF(t)\over dt}\right)+ (t-3) F(t)=0  \,.
\end{equation}

Two independent solutions of the Picard-Fuchs equation providing  the
real period $\varpi_r$ (chosen to be positive)  on the real line $\IR$, and the
imaginary period $\varpi_c$ (with $\Imm(\varpi_c)>0$) are given by~\cite{Beukers,haupt} 
\begin{equation}
  \label{e:complexp}
  \varpi_c(t)={2\pi\over (t-3)^{\frac14}(t^3-9t^2+3t-3)^{\frac14}}\,
  {}_2F_1\left({{1\over12},{5\over12}\atop 1},{1728\over J(t)}\right)\,,
\end{equation}
and 
\begin{equation}
  \label{e:realp}
  \varpi_r(t)={12\sqrt3\over t-9} \,\varpi_c\left(9{t-1\over t-9}\right)\,.
\end{equation}
This real period can as well be computed by the following integral
\begin{equation}\label{e:realpint}
    \varpi_r(t) =6\int_{-1}^0  {dx\over {\partial
      \As(x,y)\over \partial y}|_{\cE}}= 6\int_{-1}^0  {dx\over \sqrt{ (1 + (3 - t) x + x^2)^2 -4 x (1 + x)^2 }}
\end{equation}

Using the result of Maier in~\cite{haupt}, we  can derive the
expression for the  
$\Gamma_1(6)$ invariant Hauptmodul, $t$ as modular form function of
$q:=\exp(2i\pi \tau)$ with $\tau$ the period ratio $\tau=\varpi_c/\varpi_r$
\begin{equation}\label{e:thaupt9}
t=  9+72{\eta(q^2)\over\eta(q^3)}\,\left(\eta(q^6)\over\eta(q)\right)^5\,.
\end{equation}
The four cusps of $\Gamma_1(6)$ located  $\tau=0,1/2,1/3,+i\infty$ are
mapped to the values  $t=+\infty,1,0, 9$ respectively.

Plugging the expression for $t$ into the one for the periods
in~\eqref{e:complexp} and~\eqref{e:realp}, and performing the $q$-expansion
with {\tt Sage}~\cite{sage}, one easily identifies their
expression as a modular form using the {\tt Online encyclopedia of
  integer sequences}~\cite{oeis}  
\begin{equation}\label{e:periodq}
  \varpi_r={\pi\over\sqrt3}\,
  {\eta(q)^6\eta(q^6)\over\eta(q^2)^3\eta(q^3)^2}\,.
\end{equation}
%

%------------------------------------------------------------------------
\section{The elliptic dilogarithm for the sunset integral}\label{sec:PF}
%\section{The inhomogeneous Picard-Fuchs equation}\label{sec:PF}

The sunset integral is not annihilated by  the Picard-Fuchs
for our elliptic curve 
given in~\eqref{e:PFt2}, but satisfies an inhomogenous  Picard-Fuchs
equation 
\begin{equation}\label{e:PFinht}
  {d\over dt}\left(t(t-1)(t-9){d\over dt}\cJs^2(t)\right)+(t-3)\cJs^2(t)=-6\,.  
\end{equation}

This differential equation has been derived in~\cite{Laporta:2004rb} using the properties of the
Feynman integrals or in~\cite{MullerStach:2011ru} using cohomological methods.

\medskip
We provide here another derivation of this differential equation.
We rewrite the sunset integral in eq.~\eqref{e:Js} as 
\begin{equation}
  \cJs^2(t) =\int_0^\infty  \int_0^\infty \,{1\over (1+x+y)(1+{1\over x}+{1\over y}) -t}   \, {dxdy\over xy}
\end{equation}
 Since $x+x^{-1}\geq2$ for
$x>0$,  
\begin{equation}
 (1+x+y)(1+{1\over x}+{1\over y})\geq 9, \quad\forall x\geq 0, y\geq 0 
\end{equation}
This implies that the integral has a branch cut  for $t>9$ and the integral is analytic for
$t\in\IC\backslash [9,+\infty[$.
The value $t=9$ corresponds to the three-particle threshold when all the internal lines
are cut and the integral has a logarithmic singularity
$\cJs^2(t)\propto \log(9-t)$ for $t\sim 9$.

For $t<9$ we can perform the series expansion 
\begin{eqnarray}
  \label{e:series}
  \cJs^2(t) &=& \sum_{n\geq0}  I_{n} \, t^n\cr
I_n&:=& \int_0^\infty \int_0^\infty\,  ( (1+x+y)(1+x^{-1}+y^{-1}))^{-n-1} {dxdy\over xy}\cr
&=&{1\over n!^2}\,\int_{[0,\infty]^4}\, u^{n+1} \, v^{n+1}\, e^{-u (1+x+y)- v(1+x^{-1}+y^{-1})}
{dxdydudv\over xy u v}\cr
 &=&{1\over  4^{n-1}\,n!^2}\, \int_0^\infty x^{2n+1} K_0(x)^3\,
  dx\,.
\end{eqnarray}

Resumming the series we have the following integral representation in
terms of Bessel functions
\begin{equation}\label{e:Jbessel}
  \cJs^2(t)  = 4\int_0^\infty\,x I_0(\sqrt{t} x) \, K_0(x)^3\, dx \qquad  |t|<9\,.
\end{equation}
For the generic case of three different masses it is immediate to
derive using the same method, the following representation for the sunset integrals for $K^2< (m_1+m_2+m_3)^2$  
\begin{equation}
  \cIs^2(K^2,m_1,m_2,m_3)=4\pi^2\mu^2\int_0^\infty \,x
  I_0(\sqrt{K^2}x) \,\prod_{i=1}^3 K_0(m_ix) dx\,.
\end{equation}
These integrals  are  particular cases of the Bessel moment 
 integrals  discussed in~\cite{Bailey:2008ib,Broadhurst:2008mx}. 
Using the $t$ expansion in eq.~\eqref{e:series},  we 
obtain the following  action of the Picard-Fuchs
operator
\begin{multline}
     {d\over dt} \left( t(t-1)(t-9)  {d\cJs^2(t)\over dt}\right)+
     (t-3) \cJs^2(t)
=-3 I_0+9I_1\cr
+ \sum_{n\geq1} \, (n^2 I_{n-1}-(3+10n+10n^2)\,I_n+9(n+1)^2I_{n+1})\, t^n\,.
\end{multline}
It is easy to show that $I_1=(\cJs^2(0)-2)/3$ therefore  $-3I_0+9I_1=-6$.
Introducing 
\begin{equation}
\hat I_{n}:={1\over 4^{n-1}n!^2}\int_0^\infty \, x^{2n+1}\,K_0(x) K_1(x)^2\,dx\,,
\end{equation}
so that
$\hat I_1=1/3$ and $\hat I_n$ is convergent
for $n>0$.
Setting $v_n:=
\begin{pmatrix}
  I_n\cr \hat I_n
\end{pmatrix}$
a repeated use of integration by parts identities gives the matrix relation for $n>1$
\begin{eqnarray}
  v_{n+1}
&=&{1\over 9^{n}\,(n+1)!^2}\,
M(n)\cdot M(n-1)\cdots M(1)\,v_1\cr
M(n)&:=&\begin{pmatrix}
 (3+7n)(n+1) &-6n^2 \cr
 -2n(n+1) & 3n^2
\end{pmatrix}\,.
\end{eqnarray}
Using this expression one finds that for $n\geq1$
\begin{equation}
     n^2 v_{n-1}-(3+10n(1+n))\,v_n+9(n+1)^2v_{n+1} =
     \begin{pmatrix}
       0&0 \cr -{2n(n+1)\over 9^{n-1} n!^2}& {2n^2-1\over 9^{n-2} n!^2}
     \end{pmatrix} \, v_{n-1}\,.
\end{equation}
This implies that for $n\geq1$ 
\begin{equation}
n^2 I_{n-1}-(3+10n+10n^2)\,I_n+9(n+1)^2I_{n+1}=0
\end{equation}
which is one of the Ap\'ery-like recursions considered by Zagier
in~\cite{ZagierApery}.

Therefore for $|t|<9$ the sunset integral $\cJs^2(t)$ satisfies the
differential equation in eq.~\eqref{e:PFinht}.

%-------------------------------------------------------------------------
\subsection{Solving the picard-fuchs equation}
\label{sec:solving-picard-fuchs}

We turn to the resolution of the differential equation~\eqref{e:PFinht} for the sunset
integral for $t\in\mathbb C\backslash [9,+\infty[$.

A particular solution to the inhomogeneous equation is given by a
direct application of the Wronskian method  to which we add solutions
to the homogeneous solution in $\varpi_c(t)\,\IC+\varpi_r(t)\,\IC$
(see~\cite[chap~XII]{Dieudonne} for a general reference on this method)
\begin{equation}\label{e:solpf1}
  -{\cJs^2(t)\over 6} =\alpha \,\varpi_c(t) +\beta\,\varpi_r(t) +{1\over12\pi}\,
  \int_0^t  (\varpi_c(t) \varpi_r(x) -\varpi_c(x) \varpi_r(t))\, dx\,  
\end{equation}
where $\alpha$ and $\beta$ are constant that will be determined
later. The action of the Picard-Fuchs operator on this expression gives
\begin{equation}
  \label{e:PFsol1}
  L_t \left( -{\cJs^2(t)\over 6} \right)= {t(t-1)(t-9)\over 12\pi}\,  (\varpi_c'(t) \varpi_r(t) -\varpi_c(t) \varpi_r'(t))\,.
\end{equation}
The quantity $W(t):= \varpi_c'(t) \varpi_r(t) -\varpi_c(t)
\varpi_r'(t)$ is the Wronskian for the second order differential
equation in~\eqref{e:PFt2}, therefore $t(t-1)(t-9) W(t)$ is a
  constant. Using the expression for the periods given in the previous
  section we find that this constant is  equal to
  $12\pi$ and therefore that $  L_t \left( -{\cJs^2(t)\over 6} \right)= 1$. This shows that $\cJs^2(t)$ satisfies the inhomogeneous equation~\eqref{e:PFinht}.

\medskip
Changing variables
from $x$ to $q:=\exp(2i \pi\tau(x))$ and using 
 the relation given by Maier in~\cite{haupt}
\begin{equation}
  \label{e:toToq}
  {1\over 2i \pi}{dt\over d\tau}={dt\over d\log q}= 72\,{\eta(q^2)^8\eta(q^3)^6\over\eta(q)^{10}}\,.  
\end{equation}
and finally using the expression in terms of $\eta$-function in~\eqref{e:periodq},
one can rewrite the integral as 
\begin{multline}
 -{\cJs^2(t)\over 6}=   \alpha\,\varpi_c(t)+\beta\, \varpi_r(t)+2\sqrt3\,\varpi_c(t)\, \int_{-1}^{q} \,
  {\eta(\hat q^6)\eta(\hat q^2)^5\eta(\hat q^3)^4\over \eta(\hat q)^4}\,d\log \hat q\cr
+{i\sqrt3\over\pi}\,\varpi_r(t)\,
 \,\int_{-1}^{q} \, {\eta(\hat q^6)\eta(\hat q^2)^5\eta(\hat q^3)^4\over
   \eta(\hat q)^4}\,\log \hat q\, d\log \hat q\,.
\end{multline}
Integrating by part the  second line leads to
\begin{equation}\label{e:Int1}
 -{\cJs^2(t)\over 6}=    \alpha\,\varpi_c(t)+\beta\, \varpi_r(t)
+\varpi_r(t)
 \,\int_{-1}^{q} \,L(\hat q) \,d\log \hat q\,,
\end{equation}
where we have introduced 
\begin{equation}\label{e:Lqdef}
 L(q):= -{i\sqrt3\over \pi}\int_{-1}^{q} \,
  {\eta(\hat q^6)\eta(\hat q^2)^5\eta(\hat q^3)^4\over \eta(\hat q)^4}\,d\log
  \hat q\,.
\end{equation}
For evaluating this expression we start by performing the
$q$-expansion of the integrand
\begin{equation}
 {\eta(q^6)\eta(q^2)^5\eta(q^3)^4\over \eta(q)^4}= \sum_{k\geq0} k^2
 \, \left({q^k\over 1+q^k+q^{2k}}+ {q^{2k}\over 1+q^{2k}+q^{4k}}\right)\,,
\end{equation}
which we integrate term by term using the result of the integral for $|X_0|, |X|\leq1$
\begin{equation}
  \int_{X_0}^X \left({x\over 1+x+x^2}+{x^2\over 1+x^2+x^4}\right)\,d\log x= {i\over\sqrt3}\, (f(X)-f(X_0))
\end{equation}
with 
\begin{equation}
  f(x):=\tanh^{-1}\left(x\over\zeta_6\right)-\tanh^{-1}\left(x\over\bar\zeta_6\right)\,,
\end{equation}
where $\zeta_6:=\exp(i\pi/3)=1/2+i\sqrt{3}/2$ is a sixth root of unity.
Since  $f(1)=-f(-1)=-i\pi/2$,  using that $\zeta(-1)=-1/12$, we have 
\begin{equation}
   {i\over\pi}\, \sum_{k\geq0} k \,  f((-1)^k)=-{i\over\pi}\,f(1)\,3\zeta(-1)={1\over8}\,.
\end{equation}
Therefore $L(q)$ in eq.~\eqref{e:Lqdef} has the $q$-expansion
\begin{equation}
L(q)={1\over 8}+ { 1\over\pi}\, \sum_{k\geq0} k \, \left(\tanh^{-1}\left(q^k\over\zeta_6\right)-\tanh^{-1}\left(q^k\over\bar\zeta_6\right)\right) \,.
\end{equation}
For evaluating the integral in~\eqref{e:Int1}, we integrate  this series term by term, using now that, for $|X_0|, |X|\leq1$, 
\begin{equation}
  \int_{X_0}^X \left(\tanh^{-1}\left(x\over\zeta_6\right)-\tanh^{-1}\left(x\over\bar\zeta_6\right)\right)\,
  d\log x= h(X)- h(X_0)  
\end{equation}
where
\begin{equation}\label{e:hdef}
  h(x):= \frac i2\,\left(\Li2(x\zeta_6^5)
    +\Li2(x\zeta_6^4)-\Li2(x\zeta_6^2) -\Li2(x\zeta_6)\right)\,.
\end{equation}
Since $h(1)=-h(-1)= {5\over 4\sqrt3}\,\cJs^2(0)$ obtained using~\eqref{e:I0li2}, we   therefore find that
\begin{equation}
  \int_{-1}^{q} L(\hat q)\, d\log\hat   q={1\over8} \,  \log(-q) 
  +{1\over\pi}\sum_{k\geq0} h(q^k)\,.
\end{equation}
We find that the solution to the inhomogeneous Picard-Fuchs equation
in eq.~\eqref{e:solpf1}
is given by
\begin{equation}
 -{\cJs^2(t)\over6} = \alpha'\,\varpi_c(t)+ \beta'\varpi_r(t)+
  {\varpi_r(t)\over\pi} \,  E_\suns(q) \,.
\end{equation}
where we introduced $E_\suns(q)$ defined by
\begin{eqnarray}\label{e:Hsunset}
    E_\suns(q) &:=&\sum_{n\geq0} h(q^n)- {h(1)\over2}\cr
\nn &=&- {1\over2i} \sum_{n\geq0}  \left(\Li2(q^n\zeta_6^5)
    +\Li2(q^n\zeta_6^4)-\Li2(q^n\zeta_6^2) -\Li2(q^n\zeta_6)\right)\\
&+&{1\over 4i} \, \left(\Li2(\zeta_6^5)
    +\Li2(\zeta_6^4)-\Li2(\zeta_6^2) -\Li2(\zeta_6)\right)\,.
\end{eqnarray}
Matching  the values of the sunset integral at $t=0$ and $t=1$ we find
that $\alpha'=i\pi/3$ and $\beta'=-i\pi/6$ 
\begin{equation}
 -{\cJs^2(t)\over 6}=-i{\pi\over6}\,(-2\varpi_c(t)+\varpi_r(t))
+ {\varpi_r(t)\over\pi} \,E_\suns(q) \,. \label{e:amplitude}
\end{equation}

This shows that  the one-mass sunset integral in two dimensions is expressed
in term of the elliptic dilogarithm $E_\suns(q)$ with the real period
$\varpi_r$ and complex periods
$\varpi_c$ respectively defined in~\eqref{e:realp} and~\eqref{e:complexp}.

\medskip
For obtaining the $q$-expansion of this elliptic dilogarithm we first
rewrite $h(x)$ in eq.~\eqref{e:hdef} as 
\begin{equation}
  h(x)=\sum_{n\geq1} \psi(n)\, {x^n\over
  n^2}
\end{equation}
where we have set
\begin{equation}
  \label{e:psidef}
  \psi(n):={(-1)^{n-1}\over\sqrt3}\, \left(\sin({n\pi\over3})+\sin({2n\pi\over3})\right)
\end{equation}
so that $\psi(n)$ is a character such
that $\psi(n)=1$ for $n=1\mod 6$, and $\psi(n)=-1$ for
$n=5\mod 6$,   and zero otherwise.
Then, 
\begin{equation}
E_\suns(q)= \sqrt3\, \sum_{n\geq0\atop k\geq1}
  {\psi(k)\over k^2}\, q^{nk} - {\sqrt3\over2}\, \sum_{k\geq1} {\psi(k)\over k^2}\,,
\end{equation}
summing over $n$ and using that $\psi(-k)=-\psi(k)$ we have
\begin{equation}
  E_\suns(q)={\sqrt3\over2}\,\sum_{k\in \mathbb Z\backslash \{0\}}{\psi(k)\over k^2}\, {1\over
  1-q^k}\,.
\end{equation}

We remark that the expression for this elliptic dilogarithm changes
sign under the transformation $q\to 1/q$. 

In section~\ref{sec:eis} we will show how this object can be derived from the motive
associated with the elliptic curve. See  in particular eq.~\eqref{e:ampmot} and the argument preceeding that equation.

Finally we notice that the invariance of the Hauptmodul $t$
in~\eqref{e:thaupt9} under $\Gamma_1(6)$ implies the following transformations of the elliptic dilogarithm sum
\begin{eqnarray}
 \tau'= \tau+1 &:& E_\suns(q')  = E_\suns(q) -{i \pi^2\over3}\cr
\tau'= {5\tau-1\over 6\tau-1} &:& E_\suns(q')  ={ E_\suns(q)\over -6\tau+1} -{i \pi^2\over3}\,
 {3\tau-1\over 6\tau-1}\cr
\tau'= {7\tau-3\over12\tau-5} &:& E_\suns(q')  = {E_\suns(q)\over 12\tau-5}\,.
\end{eqnarray}
In particular, the elliptic dilogarithm $E_\suns(q)$ is not modular invariant. 

%------------------------------------------------------------------------------------   
 \section{The motive for the sunset graph: Hodge structures}\label{section:motiveHS}
 
The next three sections are devoted to the motivic calculation of the amplitude. 

Let $H=H_\Q$ be a finite dimensional $\Q$-vector space. A {\it pure
  Hodge structure} of weight $n$ on $H$ is a decreasing filtration
(Hodge filtration) $F^*H_\C$ on $H_\C= H_\Q\otimes \C$ such that
writing $\overline{F}^p$ for the complex conjugate filtration and
$H^{p,q}:= F^p\cap \overline{F}^q$, we have \eq{1}{H_\C = \bigoplus_p
  H^{p,n-p}.  }

A {\it mixed Hodge Structure} on $H$ is a pair of filtrations \newline\noindent
(i) $W_*H_\Q$ an increasing filtration (finite, separated, exhaustive) called the weight filtration.\newline\noindent
(ii) $F^*H_\C$ a decreasing filtration (finite, separated, exhaustive), the Hodge filtration, defined on $H_\C=H_\Q\otimes_\Q \C$. Note that the graded pieces for the weight filtration $gr^W_nH:= W_nH/W_{n-1}H$ inherit a Hodge filtration
\ml{2}{F^p(gr^W_nH)_\C = (F^pH_\C\cap W_{n,\C})/(F^pH_\C\cap W_{n-1,\C}) \cong \\
(F^pH_\C\cap W_{n,\C}+W_{n-1,\C})/W_{n-1,\C}.
}
For a mixed Hodge structure we require that with this filtration $gr^W_nH$ should be a pure Hodge structure of weight $n$ for all $n$. \nnn
{\bf Fundamental Example.} The notion of mixed Hodge structure has
extremely wide application. All Betti groups of algebraic varieties
carry canonical and functorial mixed Hodge structures~\cite{deligneHII,deligneHIII,bertin}. Here ``Betti group'' means Betti cohomology, relative cohomology, cohomology with compact supports, cohomology of diagrams, homology, etc. ``Functorial'' means functorial for algebraic maps. 
\begin{exs}\label{exs2.2} (i) $H^1(C,\Q)$ is a pure Hodge structure of weight $1$, where $C$ is a compact Riemann surface. This is just the statement that cohomology classes with $\C$-coefficients can be decomposed into types $(1,0)$ and $(0,1)$. More generally, the Hodge decomposition for cohomology in degree $p$ for any smooth projective variety $X$ and any degree $p$, says that $H^p(X,\Q)$ is a pure Hodge structure of weight $p$. \newline\noindent
(ii) $\Q(n)$ is the pure Hodge structure of weight $-2n$ with underlying vector space a $1$-dimensional $\Q$-vector space and Hodge structure $\Q(n)_\C = \Q(n)_\C^{-n,-n}$.
\end{exs}

The categories of pure and of mixed Hodge structures are abelian. In
the mixed case this is a real surprise, because categories of filtered
vector spaces do not tend to be abelian. Let $V$ and $W$ be filtered
vector spaces, and let $\phi: V \to W$ be a linear map compatible with
the filtrations. The image of $\phi$, $\Imm(\phi)$, has in general two
possible filtrations; one can take $fil_n\Imm(\phi)$ to be either the
image of $fil_nV$ or else the intersection $\Imm(\phi)\cap
fil_nW$. Technically, the image and coimage differ in the category of
filtered vector spaces, and this prevents the category from being
abelian. In the case of mixed Hodge structures, the two filtrations can be seen to coincide. Of particular importance for us is that an exact sequence of mixed Hodge structures induces exact sequences on the level of $W_n$ and on the level of $F^p$ for all possible $n, p$. 

Categories of Hodge structures carry natural tensor product structures induced by ordinary tensor product on the underlying vector spaces with the usual notion of tensor products of filtrations. For example if $H$ is a Hodge structure, then $H(1):= H\otimes \Q(1)$ is the Hodge structure with underlying vector space $H$ with weight and Hodge structures given by $W_n(H(1)) = W_{n+2}H$ and $F^pH(1)_\C = F^{p+1}H$.

Although we do not explicitly refer to polarizations, we will assume the weight-graded pieces of our Hodge structures admit polarizations, so the category of pure Hodge structures is semi-simple.

\begin{ex} We shall need certain extensions of Hodge structures; in particular we shall need to understand $\text{Ext}^1_{MHS}(\Q(0),H)$ for a given Hodge structure $H$. Let 
\eq{27}{0 \to H \to \E \to \Q(0) \to 0
}
be such an extension. The interesting case is when $H = W_{<0}H$, in other words when the weights of $H$ are $<0$, so we assume that. (Exercise: show in general that the inclusion $W_{<0}H \subset H$ induces an isomorphism $\text{Ext}^1_{MHS}(\Q(0), H) \cong \text{Ext}^1_{MHS}(\Q(0), W_{<0}H)$.) This then determines the weight filtration on $\E$, namely $W_i\E=W_iH$ for $i<0$ and $W_0\E=\E$. It remains to understand the Hodge filtration. Let $1\in \Q(0)$ be a basis, and let $s: \Q(0) \to \E_\Q$ be a vector space splitting, so we can identify $\E_\Q=H_\Q\oplus \Q\cdot s(1)$ as a vector space. For the Hodge filtration we have
\eq{28}{0 \to F^pH_\C \to F^p\E_\C \to F^p\Q(0)_\C\to 0
}
from which one sees that $F^p\E_\C = F^pH_\C$ for $p\ge 1$ and
$F^p\E_\C = F^pH_\C+F^0\E_\C$ for $p\le 0$. Thus, the only variable is
$F^0\E_\C$. Since $F^0\E_\C \surj \Q(0)_\C$ there exists $h\in H_\C$
with $h-s(1) \in F^0\E_\C$. Clearly $h$ is well-defined up to the choice of splitting $s$ and an element in $F^0H_\C$. 
It follows that
\eq{29}{\text{Ext}^1_{MHS}(\Q(0),H)\cong H_\C/(H_\Q+F^0H_\C).
}
\end{ex}

To understand how Hodge structures are related to amplitudes, consider the dual $M^\vee$ of a Hodge structure $M$. We have
\eq{30}{M^\vee := \text{Hom}(M, \Q);\quad W_nM^\vee = (W_{-n-1}M)^\perp;\quad F^pM^\vee_\C = (F^{-p+1}M_\C)^\perp.
}
The dual of \eqref{27} reads
\eq{}{0 \to \Q(0) \to \E^\vee \to H^\vee \to 0
}
For convenience we write $M=\E^\vee$ so we have an inclusion of Hodge structures $\Q(0) \inj M$. Let $0=F^{p+1}M_\C \subset F^pM_\C\neq (0)$ be the smallest level of the Hodge filtration, and suppose we are given $\omega \in F^pM_\C$. Dualizing we get $M^\vee \surj \Q(0)$. We choose a splitting $s: \Q(0) \inj M^\vee$. The period is then
\eq{31a}{\langle\omega,s(1)\rangle \in \C.  
}
Because the Hodge filtration is not defined over $\Q$, the period
 is not rational.

The period in~\eqref{31a} is the Feynman integral or amplitude 
for instance in eq.~\eqref{e:Js} for the one mass sunset graph.
In the precise relation to the Feynman integral or amplitude, 
one needs to pay attention to the important fact  that the period in~\eqref{31a} depends on the choice
of splitting $s$. This may seem odd to the physicist, who doesn't
think of amplitude calculations as involving any choices. Imagine,
however, that the amplitude depends on external momenta. The resulting
function of momenta is generally multiple-valued. Concretely, as
momenta vary, the chain of integration, which is typically the
positive quadrant in some $\R^{4n}$, has to deform to avoid the polar
locus of the integrand.
  When the momenta return to their original
values, the chain of integration may differ from the original
chain. Thus the choice of section $s$ is a reflection of the
multiple-valued nature of the amplitude. This is a crucial point in
identifying periods with Feynman integrals.  
\begin{ex}Let $E$ be an elliptic curve and suppose $M$ arises as an
  extension  
\eq{32}{0 \to \Q(0) \to M \to H^1(E,\Q(-1)) \to 0. 
}
(We will construct a family of such extensions in a minute.) The dual
sequence\footnote{Note that $H^1(E,\Q)^\vee \cong H^1(E,\Q(1))$} is
\eq{33a}{0 \to H^1(E,\Q(2)) \to M^\vee \to \Q(0) \to 0.
}
We have $F^2M_\C \cong F^2H^1(E,\C)(-1) \cong F^1H^1(E,\C) =
\Gamma(E,\Omega^1)$ so the choice of a holomorphic $1$-form $\omega$
determines an element in the smallest Hodge filtration piece
$F^2M_\C$. In this case, we can understand the relation between the
amplitude~\eqref{31a} and the extension class
\eqref{35} as follows. We have $F^0H^1(E,\C)(2)=F^2H^1(E,\C)=(0)$ so
$F^0M^\vee_\C \cong F^0\C(0)$ and there is a canonical lift $s_F$ of
$1\in \Q(0)$ to $F^0M^\vee_\C$. The class of $M^\vee \in
Ext^1_{MHS}(\Q(0),H^1(E,\Q(2)))\cong H^1(E,\C(2))/H^1(E,\Q(2))$ is
given by $-s_F+s(1)$. The pairing $\langle\rangle:M\times M^\vee \to
\Q=\Q(0)$ can be viewed as a pairing of Hodge structures, so $\langle
\omega, s_F\rangle \in F^2\C(0) = (0)$. With care, the $\Q$ can be
replaced by $\Z$, and the amplitude becomes a map  
\ml{34a}{\langle\omega,?\rangle : Ext^1_{MHS}(\Z(0),H^1(E,\Z(2))) \cong \\
H^1(E,\C)/H^1(E,\Z(2)) \to \C/\langle\omega,H^1(E,\Z(2))\rangle.
}
(Here $\langle\omega,H^1(E,\Z(2))\rangle \subset \C$ is the lattice spanned by the image in $\C\cong H^2(E,\C(2))$ under cup product.) We have seen in section~\ref{sec:PF} for the sunset all equal masses case with mass $m$
and external momentum $K$, the amplitude is a solution of an
inhomogeneous Picard-Fuchs equation in eq.~\eqref{e:PFinht} for a family of elliptic curves,
where the parameter of the family is $t:= K^2/m^2$.
 The quotient
$\langle\omega,H^1(E,\Z(2))\rangle$ reflects the fact that the
inhomogeneous solution is multiple-valued with variation given by a
lattice of periods of the elliptic curves.  

\medskip
The amplitude is closely related to the {\it regulator} in arithmetic
algebraic geometry~\cite{soule}. Let $conj: M_\C \to M_\C$ be the real
involution which is the identity on $M_\R$ and satisfies
$conj(c\,m)=\bar c\,m$ for $c\in \C$ and $m\in M_\R$. With notation as
above, the extension class $s(1)-s_F\in H^1(E,\C)$ is well-defined
up to an element in $H^1(E,\Q(2))$ (i.e. the choice of $s(1)$). Since
$conj$ is the identity on $H^1(E,\Q(2))$, the projection onto the
minus eigenspace $(s(1)-s_F)^{conj=-1}$ is canonically defined. The
regulator is then 
\eq{35}{\langle\omega,(s(1)-s_F)^{conj=-1}\rangle \in \C. 
}
\end{ex}

Let $E\subset \P^2$ be the elliptic curve defined by the equal mass
sunset equation~\eqref{e:Asunsett}. Write $X, Y, Z$ for the
homogeneous coordinates, and let $x=X/Z, y=Y/Z$. We construct an
extension of type~\eqref{33a} as follows. The curve passes through the
points $(1,0,0), (0,1,0), (0,0,1)$. Let $\rho: P \to \P^2$ be the blowup of $\P^2$ at these three points. The inverse image of the coordinate triangle in $P$ is a hexigon we call $\mathfrak h$. The union of the three exceptional divisors will be denoted $D=D_1\cup D_2\cup D_3$. The elliptic curve lifts to $E\subset P$, meeting each edge of $\mathfrak h$ in a single point. The picture is figure~\ref{fig:sunsetfig3}. 

\begin{lem}The localization sequence
\eq{36e}{0 \to H^2(P,\Q(1))/\Q\cdot [E] \to H^2(P-E,\Q(1)) \to H^1(E,\Q) \to 0
}
is exact and canonically split as a sequence of $\Q$-Hodge structures. Here $[E] \in H^2(P, \Q(1))$ is the divisor class. 
\end{lem}
\begin{proof} Let $q\in E$ be a point of order $3$, e.g. $x=0, y=-1$. Then $q$ is a flex point, so there exists a line $L\subset \P^2$ with $L\cap E = \{q\}$. Write $S:=\{q,(1,0,0), (0,1,0),(0,0,1)\}$. Points in $S$ are torsion on $E$. It is convenient to work with the dual of~\eqref{36e} which is the top row of the following diagram \minCDarrowwidth.1cm \begin{tiny}
\eq{37}{\begin{CD}0 @>>> H^1(E)(1) @>>> H^2(P,E)(1) @>>> H^2(P)^{(0)}(1) @>>> 0 \\
@. @VVV @VV a V @VV 0 V \\
0 @>>> H^1(E-S)(1) @>b >> H^2(P-D-L,E-S)(1) @>>> H^2(P-D-L)(1) @>>> 0
\end{CD}
}
\end{tiny}
(Here $H^2(P)^{(0)}:= \ker(H^2(P) \to H^2(E)$.) The right hand vertical arrow is zero, so $\text{Image}(a) \subset \text{Image}(b)$. Thus, the top sequence splits after pushout along $H^1(E) \to H^1(E-S)$.  It follows that the dual sequence~\eqref{36e} arises from a map of vector spaces
\eq{38}{H^2(P,\Q(1))/\Q\cdot [E] \to Ext^1_{\text{MHS}}(H^1(E,\Q),\Q(0)) \cong E(\C)\otimes \Q,
}
and that the image of this map lies in the subgroup of $E(\C)\otimes \Q$ spanned by $0$-cycles of degree $0$ supported on $S$. These $0$-cycles are all torsion, so the extension splits. The splitting is canonical because two distinct splittings would differ by a map of Hodge structures from $H^1(E) \to H^2(P,\Q(1))/\Q\cdot [E]$ and no such map exists. 
\end{proof}

As a consequence of the lemma, we have a canonical inclusion $\iota: H^1(E,\Q(-1)) \subset H^2(P-E,\Q)$. Consider the diagram where $M$ is defined by pullback \minCDarrowwidth.1cm \begin{small}
\eq{39e}{\begin{CD} 0 @>>> H^1(\mathfrak h - E\cap \mathfrak h) @>>> H^2(P-E, \mathfrak h - E\cap \mathfrak h) @>>> H^2(P-E) @>>> 0 \\
@. @| @AAA @AA\iota A \\
0 @>>> H^1(\mathfrak h - E\cap \mathfrak h) @>>> M @>>> H^1(E,\Q(-1)) @>>> 0
\end{CD}
}
\end{small}
Dualizing the bottom sequence yields an extension of the form~\eqref{33a}. 

We continue to assume $E$ is the elliptic curve defined by~\eqref{e:Asunsett}, and we write $\omega$ for the meromorphic two-form which is the integrand in~\eqref{e:Js}. 
\begin{lem}The Feynman amplitude coincides with the amplitude $\langle
  \omega,s(1)\rangle$~\eqref{31a} associated to $M^\vee$ where $M$ is
  as above. (Here, of course, ``coincide'' means up to the ambiguity associated to the choice of section $s$.)
\end{lem}
\begin{proof}$\omega$ is a two-form on $P$ with a simple pole on $E$ and no other pole. (This is not quite obvious. The form is given on $\P^2$ and then pulled back to $P$. It is a worthwhile exercise to check that the pullback does not acquire poles on the exceptional divisors $D_1, D_2, D_3$.) It therefore represents a class in $F^2H^2(P-E,\C)$. From the previous lemma we have $H^2(P-E,\Q) \cong H^1(E,\Q(-1))\oplus \Q(-1)^3$, so $F^2H^2(P-E,\C) = F^2H^1(E,\C(-1))= F^2M_\C$. (The last identity follows from the bottom line in~\eqref{39e}.) The dual of~\eqref{39e} looks like \minCDarrowwidth.1cm \begin{small}
\eq{40a}{\begin{CD} 0 @>>> H_2(P-E) @>>> H_2(P-E, \mathfrak h - E\cap \mathfrak h) @>>> H_1(\mathfrak h - E\cap \mathfrak h) @>>> 0 \\
@. @VV\iota^\vee V @VVV  @| \\
0 @>>> H^1(E,\Q(2)) @>>> M^\vee @>>> H_1(\mathfrak h - E\cap \mathfrak h) @>>> 0.
\end{CD}
}
\end{small}
The chain of integration in~\eqref{e:Js} can be viewed as a two-chain
on $P-E$ with boundary in $\mathfrak h-E\cap \mathfrak h$. This
boundary is the loop which represents the generator of $H_1(\mathfrak
h - E\cap \mathfrak h) = \Q(0)$, so the two-chain gives a $\Q$-vector
space splitting of the top line in~\eqref{40a}. Since $\omega \in M_\C$ we can compute the pairing by pushing the chain down to $M^\vee$. But this is how the amplitude~\eqref{31a} is defined. 
\end{proof}

\begin{figure}[ht]
  \centering
  \includegraphics[width=8cm]{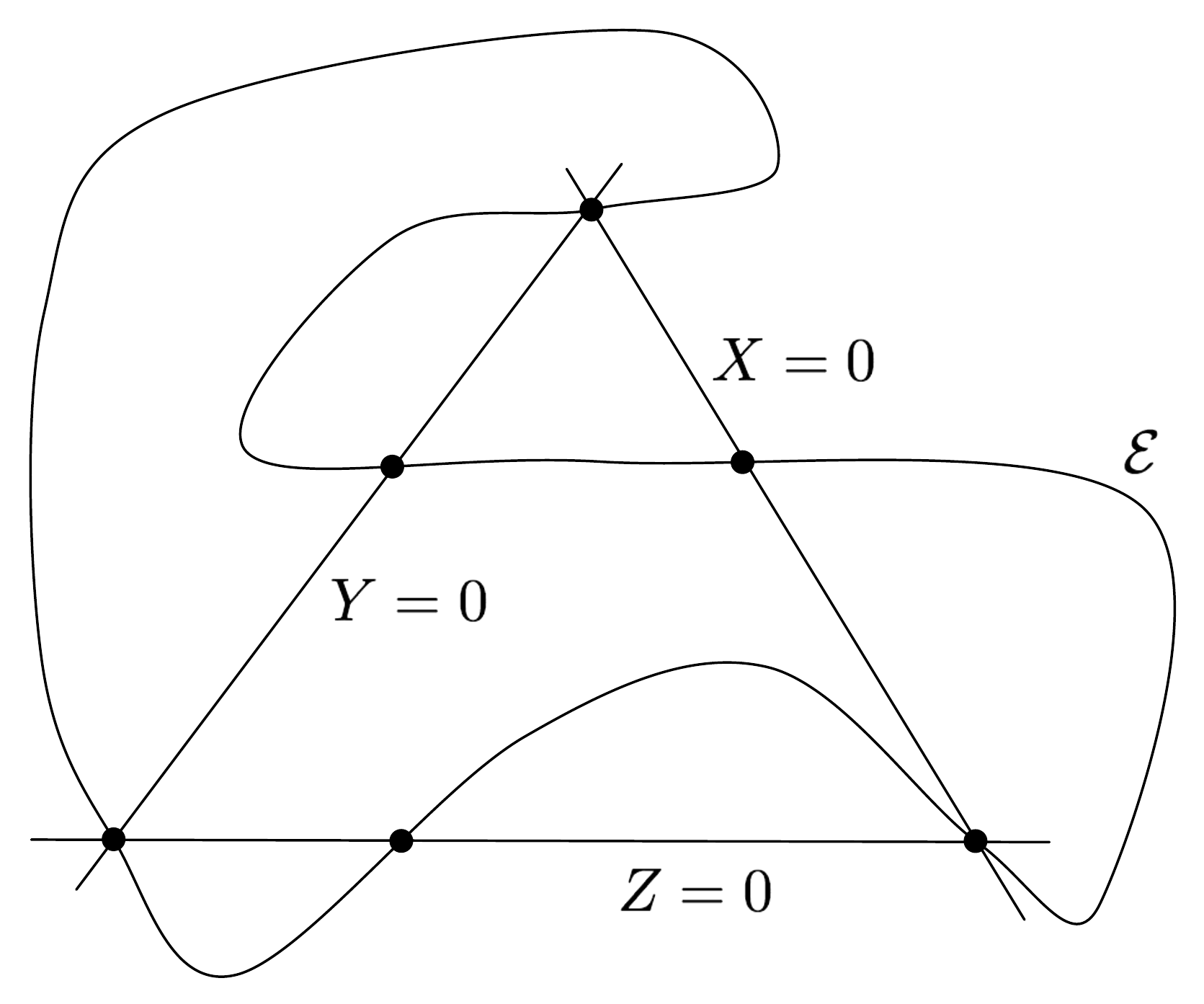}
  \caption{The sunset graph $E$ meets the coordinate triangle in six points; three corners and three other points. In the equal mass case, the difference of any two points is six-torsion on the curve.}
  \label{fig:sunsetfig2}
\end{figure}

\begin{figure}[ht]
  \centering
  \includegraphics[width=8cm]{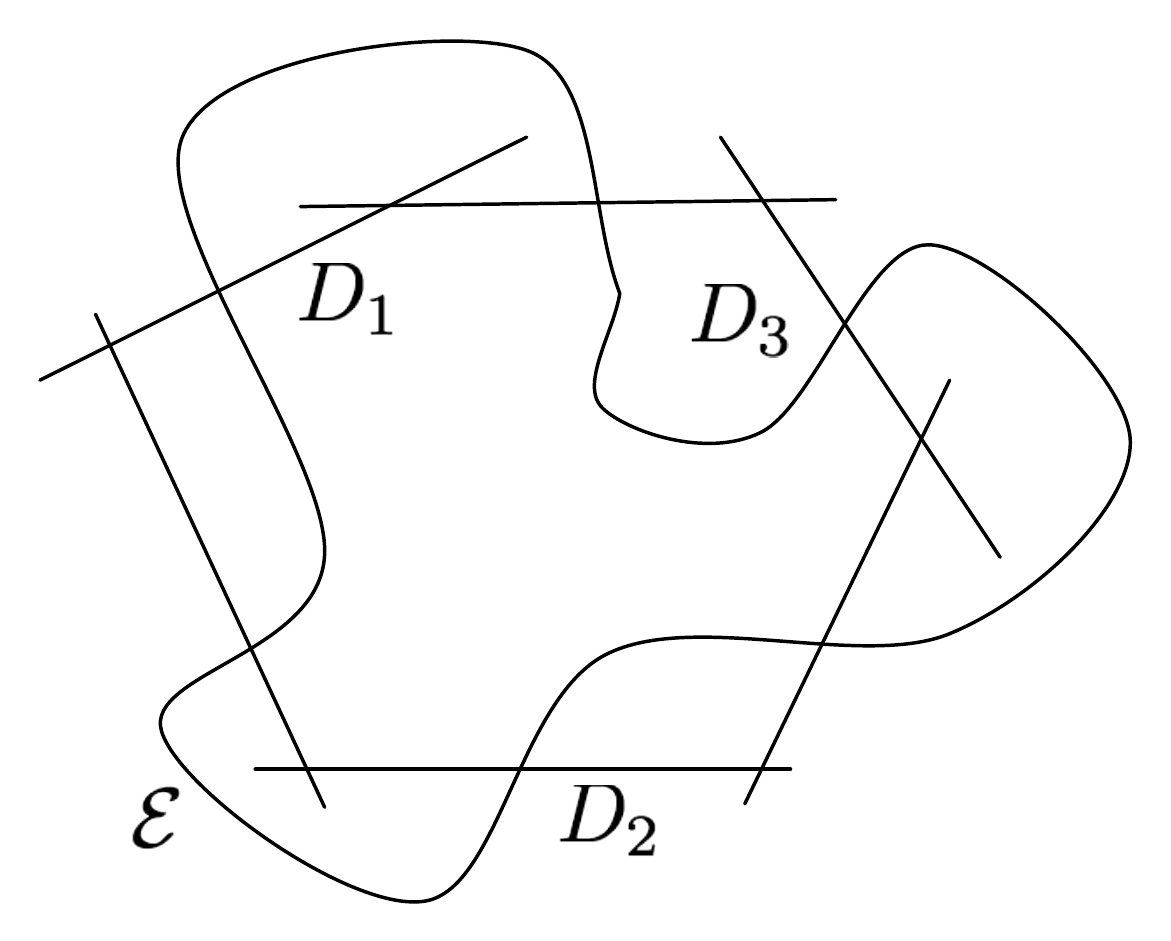}
  \caption{After blowup, the coordinate triangle becomes a hexigon in $P$ with three new divisors $D_i$. $E$ now meets each of the six divisors in one point.}
  \label{fig:sunsetfig3}
\end{figure}

Finally in this section we discuss briefly the notion of a family of Hodge strutures associated to a family of elliptic curves $f:\sE \to X$ where $X$ is an
open curve. We write $\sE_t$ for the individual
elliptic fibres. For convenience, we write $\sV = \sV_\Q:=
R^1f_*\Q_\sE(2)$ and $\sV_\C := \sV\otimes \C$. These are rank two
local systems of $\Q$ (resp. $\C$) vector spaces on $X$. We can, of course, also define the $\Z$-local system $\sV_\Z = R^1f_*\Z_\sE(2)$.

The
corresponding coherent analytic sheaf $\sV\otimes_\Q \sO_X = \sV_\C
\otimes_\C \sO_X$ admits a connection, and we may define a twisted de
Rham complex which is an exact sequence of sheaves 
\eq{36a}{0 \to \sV_\C \to \sV_\C\otimes_\C \sO_X \xrightarrow{d} \sV_\C\otimes_\C \Omega^1_X \to 0.
}

The fact that $\sV$ is a family of Hodge structures translates into a Hodge filtration $f_*\Omega^1_{\sE/X}\subset \sV\otimes \sO_X$, and in our modular case we will exhibit a section 
\eq{37a}{eis \in \Gamma(X, f_*\Omega^1_{\sE/X}\otimes\Omega^1_X)\subset \sV\otimes\Omega^1_X.  
}
\begin{remark}
The Eisenstein element $eis$ is associated to a modular family and a $\Q$-valued function on the cusps. As such, it is rather special. However, the extension of local systems which it defines via pullback \minCDarrowwidth.1cm
\eq{37d}{\begin{CD}0 @>>> \sV_\C @>>> \sV\otimes\sO_X @> d>> \sV\otimes \Omega^1_X @>>> 0 \\
@. @| @AAA @AA 1 \mapsto eis A \\
0 @>>> \sV_\C @>>> \sW @>>> \C_X @>>>0
\end{CD}
}
underlies a family of extensions of Hodge structures. If we tensor this pullback extension with $\sO_X$, we get an exact sequence of vector bundles with integrable connections 
\eq{38d}{0 \to \V \to \W \to \sO_X \to  0
}
The fibres of \eqref{38d} are of the form \eqref{33a}. 

The point is that some analogue of this sequence is available quite generally. It does not depend on the special modular nature of the sunset equal mass case. The amplitude is an algebraic section of $\W$ so the differential equation satisfied by the amplitude will be an inhomogeneous equation associated to the homogeneous equation on $\V$ with source term a rational function on $X$. 

If, for example, we consider the sunset graph with unequal masses, we get a somewhat more complicated picture. The extension of coherent sheaves with connections becomes
\eq{}{0 \to \V \to \W \to \E \to 0
}
where $\E$ itself contains Kummer extensions of the form 
\eq{}{0 \to \sO_X \to \F \to \sO_X \to 0
}
(Kummer means one has a section of $\F$, $e\mapsto 1$, and a rational function $f$ on $X$ such that 
$e+\log f\cdot 1$ is horizontal for the connection.) The source term in the inhomogeneous equation for the amplitude in this case will involve $\log f$'s. 
\end{remark}
%---------------------------------------------------------------------------------------
\section{The motive for the sunset graph: motivic cohomology}\label{sec:motive}
Consider a larger diagram of sheaves
\eq{38a}{\begin{CD}
0 @>>> \sV_\Z @>>> \sV\otimes \sO_X @>>>  (\sV\otimes \sO_X)/\sV_\Z @>>> 0 \\
@. @VVV @| @VV \delta V \\
0 @>>> \sV_\C @>>> \sV\otimes \sO_X @>d>>  \sV\otimes \Omega^1_X @>>> 0
\end{CD}
}
The sheaf $(\sV\otimes \sO_X)/\sV_\Z$ is the sheaf of germs of analytic sections of a fibre bundle over $X$ with fibre the abelian Lie group $H^1_{Betti}(\sE_t, \C/\Z(2))$.  Using {\it motivic cohomology}, we will explicit a family of extensions such that the family of extension classes give a section $mot \in \Gamma(X, \sV\otimes \sO_X/\sV_\Z)$. We will show that 
\eq{}{\delta(mot) = eis. 
}
We will relate $mot$ to the Feynman amplitude viewed as a multivalued function on $X$, and then we will use $eis$ in order to calculate this function. 

$\sE/X$ will denote the family of elliptic curves defined by the equal mass sunset equation~\eqref{e:Asunsett}. The corresponding family of extensions as in the bottom line of~\eqref{40a} determine a family of elements 
\eq{}{mot_t \in \text{Ext}^1_{MHS}(\Q(0),H^1(\sE_t,\Q(2)) \cong H^1(\sE_t, \C/\Z(2)).
}
(Compare to~\eqref{34a}.) Write $mot \in \Gamma(X, (\sV\otimes \sO_X)/\sV_\Z)$ for the corresponding section of the bundle of generalized intermediate jacobians. 
\begin{lem}\label{lem6.4} Let $\partial(mot)\in H^1(X, \sV_\Q)$ be the boundary for the top sequence in~\eqref{38a} tensored with $\Q$. We have $\partial(mot) \in \bigoplus \Q(0) \subset H^1(X,\sV_\Q)\subset H^2_{Betti}(\sE, \Q(2))$.  
\end{lem}
\begin{proof}  Consider a compactification 
\eq{39a}{\begin{CD} \sE @>>> \overline{\sE} \\
@VV f V @VV \overline f V \\
X @>>> \overline{X}
\end{CD}
}
with $\overline X$ a compact Riemann surface and fibres of $\overline f$ possibly degenerate. We assume $\overline \sE$ is a smooth variety, and we write $F=\coprod F_s = \overline\sE - \sE$. (The $F_s$ are possibly singular fibres.) Localization yields an exact sequence of Hodge structures
\eq{40b}{H^2(\overline \sE,\Q(2)) \to H^2(\sE,\Q(2)) \to H_1(F,\Q(0)) \xrightarrow{\mu} H^3(\overline \sE,\Q(2))
}
In cases of interest to us, the singular fibres $F_s$ will be nodal rational curves or a union of copies of $\P^1$ supporting a loop, so $H_1(F_s,\Q(0)) = \Q(0)$ as a Hodge structure. Since $\overline \sE$ is smooth and compact, $H^i(\overline\sE,\Q(2))$ is a pure Hodge structure of weight $i-4$. In particular, the map labeled $\mu$ in~\eqref{40b} is zero and we get an exact sequence of Hodge structures
\eq{41b}{0 \to \{\text{pure of weight -2}\} \to H^2(\sE,\Q(2)) \xrightarrow{r} \coprod_{\text{bad fibres}}\Q(0) \to 0.
}
The family of elliptic curves associated to the sunrise diagram with equal masses is {\it modular}, associated to the modular curve $X=X_1(6)$ (see section~\ref{sec:K3fam}). In the modular case, a remarkable thing happens~\cite{deligne}. The arrow $r$ in~\eqref{41b} is canonically split as a map of Hodge structures. This relates to the local system $\sV$ above because the Leray spectral sequence for $f$ yields
\eq{42}{H^1(X,\sV) = H^1(X, R^1f_*\Q_\sE(2)) \inj H^2(\sE,\Q(2)).
}
The splitting of $r$ in fact embeds $\bigoplus_{\text{bad fibres}}\Q(0) \subset H^1(X,\sV)$, and the assertion of the lemma is that $\partial(mot)$ lies in the sub Hodge structure  $\bigoplus_{\text{bad fibres}}\Q(0)$. The reason for this is that, rather than considering the section of a bundle of generalized intermediate jacobians given by the $mot_t$, one can work with the {\it motivic cohomology} of the total space $\sE$ of the family. 

Motivic cohomology $H^p_M(X,\Q(q))$ is defined using algebraic cycles on products of $X$ with affine spaces. It has functoriality properties (in the case of smooth varieties) which are analogous to Betti cohomology. For $E\subset \P^2$ an elliptic curve, we are interested in 
\eq{43}{H^2_M(E, \Q(2)) \to H^3_M(P,E;\Q(2)). 
}
Here $P$ is the blowup of $\P^2$ as in figure~\ref{fig:sunsetfig3}.  A sufficient condition for a finite formal sum $\sum_i (C_i, f_i)$ where $C_i \subset P$ is an irreducible curve and $f_i$ is a rational function on $C_i$ to represent an element in 
$H^3_M(P,E;\Q(2))$ is that firstly the sum $\sum_i (f_i)$ of the divisors of the $f_i$, viewed as $0$-cycles on $P$ should cancel, and secondly that $f_i|C_i\cap E$ should be identically $1$. As an example, take the $\mathfrak h_i$ to be the irreducible components of the hexagon $\mathfrak h\subset P$. Number the $\mathfrak h_i$ cyclically and take $f_i$ having a single zero and a single pole on $\mathfrak h_i$ such that the zero of $f_i$ coincides with the pole of $f_{i+1}$, the unique point $\mathfrak h_i\cap \mathfrak h_{i+1}$. Since $E\cap \mathfrak h_i = \{p_i\}$ is a single point which is not a zero or pole of the $f_i$, we can scale $f_i$ so $f_i(p_i)=1$. The resulting formal sum represents an element $T\in H^3_M(P,\sE;\Q(2))$. 

This element lifts back to an element $S\in H^2_M(\sE, \Q(2))$. Indeed, the Milnor symbol 
\eq{43a}{S:= \{-X/Z,-Y/Z\}
}
represents an element in the motivic cohomology of the function field of $\sE$. To check that it globalizes, we should check that there are no non-trivial tame symbols at the six points where the fibres $\sE_t$ meet the coordinate triangle. By symmetry it suffices to check the points with homogeneous coordinates $(0,-1,1)$ and $(0,0,1)$. Recalling the general formula
\eq{44}{\text{tame}_x\{f,g\} = (-1)^{ord_x(f)ord_x(g)}(f^{ord_x(g)}/g^{ord_x(f)})(x)\in k(x)^\times
}
it is straightforward to check that the tame symbols are both $\pm
1$. Since we are working with coefficients in $\Q$ the presence of
torsion need not concern us, and we conclude that $S$ globalizes. The
assertion that $S\mapsto T$ in~\eqref{43} amounts to the assertion
that $S$ viewed as a Milnor symbol now on $P$ has $T$ as tame
symbol. For more on this sort of calculation, one can see~\cite{doran}.
 
We will also need {\it Deligne cohomology} $H^p_\sD(X,\Q(q))$ (see the articles of Schneider, Esnault-Viehweg, and Jannsen in~\cite{rss}, as well as~\cite{soule}). It sits in an exact sequence 
\ml{45}{0 \to \text{Ext}^1_{MHS}(\Q(0), H^{p-1}_{Betti}(X,\Q(q))) \to H^p_\sD(X,\Q(q)) \xrightarrow{b} \\
\text{Hom}_{MHS}(\Q(0),H^{p}_{Betti}(X,\Q(q))) \to 0.
}
(The referee points out that this is true only for $X$ smooth, projective. The analogous result for more general $X$ involves the absolute Hodge cohomology whose definition involves both the weight and Hodge filtration. One has quite generally
\eq{}{Ext^1_{MHS}(\Q(0), H^{p-1}(X, \Q(q))) \cong \frac{W_{2q}H^{p-1}(X,\C)}{F^qW_{2q}H^{p-1}(X,\C)+W_{2q}H^{p-1}(X,\Q)}.
}
In our case, $p=q=2$ so $W_{2q}H^{p-1}=H^{p-1}$ and there is no
distinction.)

Motivic and Deligne cohomologies are related by the regulator map 
\eq{46}{reg: H^p_M(X,\Q(q)) \to H^p_\sD(X, \Q(q)).
}
The composition
\ml{47}{b\circ reg: H^p_M(X,\Q(q)) \to \text{Hom}_{MHS}(\Q(0),H^{p}_{Betti}(X,\Q(q))) \\
\subset H^{p}_{Betti}(X,\Q(q))
}
is the Betti realization. Note that the image lands in a sub-Hodge structure of the form $\oplus \Q(0)$. 

To finish the proof of the lemma, we consider two diagrams. \begin{small}
\eq{48}{\begin{CD} H^2_M(\sE, \Z(2)) @>>> H^2_M(\sE_t, \Z(2)) \\
@VV reg V @VV reg V \\
H^2_\sD(\sE,\Z(2)) @>>> H^2_\sD(\sE_t, \Z(2)) \\
@VV\text{localize}V @AA \cong A \\
 \Gamma(X, (\sV\otimes \sO_X)/\sV_\Z) @>\text{rest. to fibre}>> H^1_{Betti}(\sE_t, \C/\Z(2)) \\
 @. @VV \cong V \\
 @. \text{Ext}^1_{MHS}(\Z(0), H^1_{Betti}(\sE_t, \Z(2))). 
\end{CD}
}
\end{small}
The arrow that requires some comment here is labeled ``localize''. A Deligne cohomology class on $\sE$ yields by restriction Deligne cohomology classes on the fibres $\sE_t$ which are just elements in $H^1_{Betti}(\sE_t,\C/\Z(2))$. As $t$ varies, however, these classes are not locally constant. They do not glue to sections of $\sV_\C/\sV_\Z$ but rather to $(\sV\otimes \sO_X)/\sV_\Z$ as indicated. 

The second relevant commutative diagram is
\eq{49}{\begin{CD}H^2_M(\sE, \Q(2)) @> b >> H^2_{Betti}(\sE, \Q(2)) \\
@VVV @AA\text{inject} A \\
\Gamma(X, (\sV\otimes \sO_X)/\sV_\Z) @> \partial >> H^1(X, \sV_\Q) 
\end{CD}
}
Assembling these two diagrams and using that the image of $b$  in~\eqref{49} lies in the subspace of Hodge classes $\bigoplus \Q(0) \subset H^2_{Betti}(\sE, \Q(2))$, the lemma is proven. 
\end{proof}

The fact that $mot \in \Gamma(X, (\sV\otimes \sO_X)/\sV_\Z)$ comes from motivic cohomology enables us to control how this section degenerates at the cusps. This is a slightly technical point and we do not give full detail. The argument is essentially an amalgam of~\cite{rss}, p. 47 in the expos\'e of Esnault-Viehweg, where the regulator of symbols like $S$~\eqref{43a} in $H^1_{Betti}(\sE_t,\C/\Z(2))$ is calculated, and the classical computation of the Gau\ss-Manin connection~\cite{katz}. Let $X \inj \overline X$ be the compactification, and let $c\in \overline X-X$ be a cusp. Let $D^* \subset X$ be a punctured analytic disk around $c$. Consider a symbol $\{g,h\}$ on the family $\sE_{D^*}$. Using~\cite{rss}, we can identify the target of the regulator map for the family with the hypercohomology group $\H^1(\sE_{D^*}, \sO^\times_{\sE_{D^*}} \xrightarrow{d\log} \Omega^1_{\sE_{D^*}/D^*})$. For our purposes it will suffice to calculate the regulator on the open set $\sE'_{D^*}= \sE_{D^*}-\text{div}(g)-\text{div}(h)$.

Let $\{U_j\}$ be an open analytic covering of $\sE'_{D^*}$ such that $\log h$ admits a single-valued branch $\log_j h$ on $U_j$. Let $m_{jk} = (\log_k h - \log_j h)/2\pi i$ on $U_{jk}=U_j\cap U_k$. The regulator applied to $\{g,h\}$ is represented by the Cech cocycle 
\eq{50}{(g^{m_{jk}}, \frac{1}{2\pi i}\log_j h\cdot dg/g) \in C^1(\{U_j\},\sO^\times)\oplus C^0(\{U_j\}, \Omega^1_{\sE'_{D^*}/D^*})
}
We want to compute the image of this class under the map $\delta$~\eqref{38a}. Consider the diagram of complexes with exact columns \minCDarrowwidth.1cm
\eq{51}{\begin{CD}@. 0 @. 0 \\
@. @VVV @VVV \\@. \sO_{\sE'_{D^*}}\otimes \Omega^1_{D^*} @>d >> \Omega^1_{\sE'_{D^*}/D^*}\otimes \Omega^1_{D^*}  \\ 
@. @VVV @VV \cong V \\
\sO_{\sE'_{D^*}}^\times @> d\log >> \Omega^1_{\sE'_{D^*}} @> d >> \Omega^2_{\sE'_{D^*}} \\
@| @VVV @VVV  \\
\sO_{\sE'_{D^*}}^\times @> d\log >> \Omega^1_{\sE'_{D^*}/D^*} @> >> 0 \\
@. @VVV \\
@.  0
\end{CD}
}
The cocycle~\eqref{50} represents a one-cohomology class in the bottom complex. The vertical coboundary yields a two-cohomology class in the top complex, which is just the relative de Rham complex shifted and tensored with $\Omega^1_{D^*}$. Note the relative de Rham complex is exactly the bundle $\sV\otimes \sO_X$ restricted to $D^*\subset X$. In fact, this vertical coboundary is exactly the map $\delta$ from~\eqref{38a}. To calculate, we view the $0$-cochain $\frac{1}{2\pi i}\log_j h\cdot dg/g$ as a cochain in the absolute one-forms $\Omega^1_{\sE'_{D^*}}$ and apply $d$
\eq{52}{d(\frac{1}{2\pi i}\log_j hdg/g) = \frac{1}{2\pi i}dh/h\wedge dg/g \in \Omega^2_{\sE'_{D^*}} \cong \Omega^1_{\sE'_{D^*}/D^*}\otimes \Omega^1_{D^*}. 
}
Notice that a log form $dh/h\wedge dg/g$ cannot have more than a first order pole along the components of the fibre over the cusp $c$. Furthermore, it is clear from~\eqref{52} that 
\eq{}{\delta(mot) \in \Gamma(\overline X, \Omega^1_{\overline\sE/\overline X}(\log \text{cusps})\otimes \Omega^1_{\overline X}).
}
We can now rewrite the diagram from~\eqref{38a} \minCDarrowwidth.1cm
\eq{52a}{\begin{CD}\Gamma(X, (\sV\otimes \sO_X)/\sV_\Z) @> \partial >> H^1(X, \sV_\Z) \\
@VV \delta V @VVV \\
\Gamma(X, \sV\otimes \Omega^1_X) @>>> H^1(X, \sV_\C) \\
@AA i A \\
\Gamma(X, \Omega^1_{\sE/X}\otimes \Omega^1_X) 
\end{CD}
}
We conclude from~\eqref{51} that $\delta(mot)$ lies in the image of $i$ in~\eqref{52a} and further that it has at worst logarithmic poles at the cusp. It is known that the space of such sections is spanned by the Eisenstein series $eis_\psi$ discussed in section~\ref{39b}. 

One other important piece of information we have by virtue of the motivic interpretation of the amplitude concerns the behavior of the Eisenstein section $\delta(mot)$ at the cusps. This is determined by the behavior of the Milnor symbol $S$~\eqref{43a} under the tame symbol mapping
\eq{}{tame: H^2_M(\sE, \Q(2)) \to \coprod_{c\in \overline X-X} H^M_1(\sE_c,\Q(0))
}
The elliptic curve $\cE_t$ in~\eqref{e:Asunset} has four cusps at
$t=0,1,9, \infty$, discussed in section~\ref{sec:K3fam}. It is elementary to check that the tame symbol
$tame(S)$ is trivial for $t=0,1,9$. Indeed, at $t=0$ the
curve~\eqref{e:Asunsett} becomes reducible with components
$1+X/Z+Y/Z=0$ and $X/Z+Y/Z+XY/Z^2=0$. Since the entries $X/Z, Y/Z$ of
the symbol $S$ are not identically $0, \infty$ on either of these
components, the tame symbol vanishes. Similarly, at $t=1$ the curve
factors as $(X+Z)(Y+Z)(X+Y)$ so again there is no contribution. At
$t=9$ the curve has a unique singular point at $X=Z, Y=Z$ and there is
no contribution. Finally at $t=+\infty$ the fibre is $XYZ=0$; one has
$H_1^M(\{XYZ=0\},\Q(0))=\Q$ and $tame_\infty(S)$ is a generator.  

%---------------------------------------------------------------------------------------
\section{The motive for the sunset graph: Eisenstein Series}\label{sec:eis}

We now assume $X$ is a modular curve, i.e. $X\cong \H/\Gamma$ where $\H=\{\tau\in \C\ |\ \Imm(\tau)>0\}$ and $\Gamma\subset SL_2(\Z)$ is a congruence subgroup. Let $\tilde f: \widetilde\sE \to \H$ be the pullback of the family of elliptic curves. We have
\ga{38c}{\Omega^1_{\widetilde\sE/\H}=\sO_\H dz;\quad \begin{pmatrix}a & b \\ c & d\end{pmatrix}dz = \frac{dz}{c\tau+d} \\
\Omega^1_\H = \sO_\H d\tau;\quad \begin{pmatrix}a & b \\ c & d\end{pmatrix}d\tau = \frac{d\tau}{(c\tau+d)^2}.
}
Let $\psi: \Z^2 \to \C$ be a map such that
$\psi(\gamma(a,b))=\psi(a,b)$ for any $a, b\in \Z$ and any $\gamma\in
\Gamma$. Define 
\eq{39b}{eis_\psi := \sum_{\substack{(a,b)\in \Z^2\\(a,b)\neq (0,0)}}\frac{\psi(a,b)d\tau dz}{(a\tau+b)^3}.
}
It is straightforward to check that $eis_\psi$ is invariant under $\Gamma$ and descends to a section of $\Omega^1_{\sE/X}\otimes \Omega^1_X \subset \sV\otimes \Omega^1_X$ over $X=\H/\Gamma$. 

This is the classical sheaf-theoretic interpretation of Eisenstein
series for $SL_2(\Z)$~\cite{koblitz}. The next step is to determine
$\psi$ which we do by studying the constant terms of the
$q$-expansions at the cusps. We have seen in section~\ref{sec:K3fam} that for  the one-mass sunset
graph we have the $X_1(6)$ modular curve, and that the four cusps
$t=0,1,9,+\infty$ are mapped by eq.~\eqref{e:thaupt9} to the four cusps for the action of
$\Gamma_1(6)$ on $\P^1(\Q)$ represented by $\tau= 1/3,
1/2,+i\infty,0$ respectively . We know the constant term of the
$q$-expansion should vanish for $t=0,1,9$ i.e. for $\tau = +i\infty,
1/3, 1/2$. Indeed, $q=\exp(2\pi i\tau)$ is the parameter at the cusp
at $+i\infty$, so $d\tau = \frac{1}{2\pi i}dq/q$ has a pole at the
cusp. Thus, the constant term of the $q$-expansion coincides up to a
factor of $2\pi i$ with the residue at the cusp.  
\begin{lem}\label{lem4.4} For $u, v\equiv 0,1, 2, 3, 4, 5 \mod 6$, define
\eq{40c}{
eis^{u,v}(\tau) = \sum_{(a,b)\equiv (u,v)\!\!\!\! \mod 6} \frac{1}{(a\tau + b)^3}
}
We have the following assertions about constant terms of $q$-expansions. \newline\noindent
(i) The constant term of the $q$-expansion at $\tau = \infty$ vanishes unless $u\equiv 0 \mod 6$. \newline\noindent
(ii) The constant term of the $q$-expansion at $\tau = 1/2$ vanishes unless $u+2v\equiv 0 \mod 6$. \newline\noindent
(iii) The constant term of the $q$-expansion at $\tau = 1/3$ vanishes unless $u+3v\equiv 0 \mod 6$. \newline\noindent
(iv) The constant term of the $q$-expansion at $\tau =0$ vanishes unless $v\equiv 0 \mod 6$. 
\end{lem}
\begin{proof}The assertion for $\tau \to \infty$ is straightforward because the series~\eqref{40c} converges uniformly as $\tau \to +i\infty$ so we can take the limit term by term. The only terms which survive are those with $a=0$. For the other assertions, we transform the cusp in question to $+i\infty$. For example, taking $\tau' = (\tau-1)/(2\tau-1)$ transforms $\tau=1/2$ to $\tau'=+i\infty$. Substituting the inverse $\tau = (-\tau'+1)/(-2\tau'+1)$ yields
\ml{}{eis^{u,v}(\tau) = \sum_{(a,b)\equiv (u,v)\!\!\!\! \mod 6} \frac{1}{(a(-\tau'+1)/(-2\tau'+1) + b)^3} = \\
(-2\tau'+1)^3 \sum_{(a,b)\equiv (u,v)\!\!\!\! \mod 6} \frac{1}{((-a-2b)\tau'+(a+b))^3}.
}
This vanishes at $\tau' \to +i\infty$ unless $u+2v\equiv 0 \mod 6$. The arguments for (iii) and (iv) are similar. 
\end{proof}

We consider a series $\sum \psi(a,b)(a\tau+b)^{-3}$. We assume $\psi(a,b)$ only depends on $a, b\mod 6$. Define $\psi^{u,v}(a,b)$ to be $1$ if $a\equiv u,\ b\equiv v\mod 6$ and zero otherwise. Write
\eq{}{\psi = \sum_{u,v \mod 6} c(u,v)\psi^{u,v}.
}
The condition that $eis_\psi(\tau)$ should be invariant under $\Gamma_1(6)$ translates into the requirement $\psi(a,b) = \psi((a,b)\begin{pmatrix}\alpha & \beta \\ \gamma & \delta\end{pmatrix})$ for $\begin{pmatrix}\alpha & \beta \\ \gamma & \delta\end{pmatrix}\in \Gamma_1(6)$. Since 
\eq{}{(a,b)\begin{pmatrix}\alpha & \beta \\ \gamma & \delta\end{pmatrix} \equiv (a,b+\beta a) \mod 6
}
it follows that 
\eq{41c}{c(u,v) = c(u,v+\beta u);\ \ \forall \beta \in \Z.
}

Writing $(a,b) = 6^n(a',b')$ where $a', b'$ are not both divisible by
$6$, we can write $eis^{0,0}(\tau)$ as a sum with constant
coefficients of $eis^{u,v}$ with $(u,v)\neq (0,0)$. Hence we can take
$c(0,0)=0$. We can further simplify by dropping $\psi^{u,v}$ if $u,v$
have a common factor $2$ or $3$. Further we can assume $c(u,v) =
-c(-u,-v)$. We see now from (ii) in lemma~\ref{lem4.4} that terms
contributing to the constant term at the cusp $1/2$ are $(u,v) =
(4,1), (2,5)$. Since $eis^{4,1}=-eis^{2,5}$ and $c(4,1)=-c(2,5)$ there
is only one contribution which is a non-zero multiple of $c(4,1)$. It
follows that $c(4,1)=c(2,5)=0$. Similarly the terms which contribute
at $\tau = 1/3$ are $(3,1), (3,5)$. Again these are  negatives and the
same argument implies $c(3,1)=c(3,5)=0$.  

We have now the following information about the $c(i,j)$:
\ga{}{c(1,0)=c(1,1)=c(1,2)=c(1,3)=c(1,4)=c(1,5)\quad (\text{by }\eqref{41c}) \notag \\
c(2,0)=c(2,2)=c(2,4)=0\quad  (\text{by divisibility}) \notag \\
c(2,1)=c(2,3)=c(2,5)=0 \quad (\text{by~\eqref{41c} and the above}) \notag \\
c(3,1) = c(3,4) = 0 \quad (\text{by~\eqref{41c} and the above})  \\
c(3,2) = c(3,5) = 0 \quad (\text{by~\eqref{41c} and the above}) \notag \\
c(3,0)=c(3,3)=0 \quad  (\text{by divisibility}) \notag \\
c(4,x)=-c(2,-x)=0 \notag \\
c(5,x) = -c(1,0). \notag
}
For the character $\psi(n)$ defined in eq.~\eqref{e:psidef} we have now proven 
\begin{thm} \label{thm7.11}
Define 
\eq{}{\psi(a,b) = \psi(a) :={ (-1)^{a-1}\over \sqrt3}\, \left(\sin({\pi a\over3})+\sin({2\pi a\over3})\right).
} 
Then for some $c\neq 0$ we have 
\eq{}{\delta(mot) = c\cdot eis_\psi \in \Gamma(X, \sV\otimes\Omega^1_X). \ \ \ (\text{cf.~\eqref{38a}})
}
\end{thm}
\begin{proof}The list of values on the right hand side for
  $a=0,1,2,3,4,5$ is $0, 1,0,0,0,-1$. This list is uniquely determined
  up to scale by the above conditions. 
\end{proof}

Our objective now is to sharpen theorem~\ref{thm7.11} by computing the
constant $c$. The residue of the symbol $S$~\eqref{43a} at $t=+\infty$
is $1$, but we should multiply by $1/6$ because the modular curve $X_1(6)$ is ramified of order $6$ under the cusp $\tau=0$. The constant term in the
$q$-expansion of the Eisenstein series is computed in
lemma~\ref{lem7.14} below to be $-{\pi^3\over 9\sqrt{3}}$.
 Combining these values, we get 
\eq{}{\text{Residue}(S) = \frac16 = \text{Residue}(\delta(mot)) =-c\cdot \frac{-\pi^3}{9\sqrt{3}}\,,
}
therefore
\begin{equation}
  \label{e:cvalue}
  c= {-6\sqrt3\over \pi(2\pi i)^2}\,.
\end{equation}

With the following lemma we evaluate at the cusp $\tau=0$,  the constant term of the
Eisenstein series $\sum_{(a,b)\in\mathbb Z^2\atop (a,b)\neq(0,0)} \psi(a,b)/(a\tau+b)^3$ with $\psi$ the character
in~\eqref{e:psidef} 

\begin{lem}\label{lem7.14}Let $\psi(a,b)=\psi(a)$ in
  eq.~\eqref{e:psidef}  as in the theorem~\ref{thm7.11}. Then the value of the constant term of the Eisenstein series 
  \eq{}{\sum_{(a,b)\in\mathbb Z^2\atop (a,b)\neq(0,0)} {\psi(a,b)\over(a\tau+b)^3}
  } 
  at the cusp $\tau = 0$ is  $-\pi^3/(9\sqrt{3})$.
\end{lem}
\begin{proof}Let $\tau' = -1/\tau$. We need to compute
\ml{42c}{C=\lim_{\tau'\to +i\infty}\Big(\sum_{a\equiv 1\!\!\!\!\!\mod 6}\frac{1}{(-a+b\tau')^3}- \sum_{a\equiv 5\!\!\!\!\!\mod 6}\frac{1}{(-a+b\tau')^3}\Big) = \\
-2\Big(\sum_{\substack{a\ge 1\\ a\equiv 1\!\!\!\!\!\mod 6}}a^{-3}- \sum_{\substack{a\ge 1\\ a\equiv 5\!\!\!\!\!\mod 6}}a^{-3}\Big).
}
One can rewrite this constant term in terms of Hurwitz zeta functions
$\zeta(s,a):=\sum_{n\geq0} (n+a)^{-s}$  as 
\begin{equation}
  C= 2-{2\over
    6^3}\,\left(\zeta(3,\frac16)- \zeta(3,-\frac16)\right)\,.
\end{equation}
The Hurwitz zeta functions have the following $a$-expansion
\begin{equation}
  \zeta( s,a) = {1\over a^s} +\sum_{n\geq0}\, \left( s+n-1\atop n\right)\, \zeta(s+n)\, (-a)^n \,.
\end{equation}
Therefore
\begin{equation}
C=-2- {1\over54}  \sum_{n\geq0}   \, \left( 2n+3\atop 2n+1\right)\, \zeta(2n+4)\, \left(-\frac16\right)^{2n+1}\,.
\end{equation}
Remarking that

\begin{equation}
\sum_{n\geq0}   \, \left( 2n+3\atop 2n+1\right)\, \zeta(2n+4)\,
(-x)^{2n+1}= -{1\over 2x^3}+ {\pi^3\over2}\, {\cos(\pi x)\over
  \sin(\pi x)^3}
\end{equation}
one obtains that
\begin{equation}
    C=  - {\pi^3\over9\sqrt3}\,.  
\end{equation}
\end{proof}

It remains now to compute the amplitude asscociated to our Eisenstein
series $\delta(mot)= 6 \sqrt{3}\cdot eis_\psi/\pi$ and check that it
agrees with $\cJs^2(t)$ defined in eq.~\eqref{e:Js} modulo periods of the elliptic curve $\cE_t$. Consider
again the basic exact sequence~\eqref{36a}. We know from
lemma~\ref{lem6.4} that $\partial(\delta(mot)) \in H^1(X,
\sV_\Q)\subset H^1(X, \sV_\C)$. It will be convenient to pull back
this sequence and work over the upper half-plane $\H=\{\tau\in\mathbb C\ |\
\Imm(\tau)>0\}$. Write $V_\Z \subset V_\Q \subset V_\C\subset V\otimes
\sO_{\H}$ for the various base extensions of the representation of
$\Gamma_1(6)$ underlying the local system $\sV$. Let $V_\Z =
\Z\ve_1\oplus \Z\ve_2$ such that $dz = \tau\ve_1+\ve_2$. If $\gamma_1,
\gamma_2$ is the dual homology basis, this gives 
\eq{}{\int_{\gamma_1}dz = \tau;\ \ \int_{\gamma_2}dz = 1. 
}
The pairing $H^1(\sE_t,\Z) \otimes H^1(\sE_t, \Z) \to \Z(-1)$ yields a symplectic form
\eq{}{\langle \ve_1,\ve_2\rangle = 2\pi i = -\langle \ve_2,\ve_1\rangle, 
}
and $\langle \varepsilon_i,\varepsilon_i\rangle=0$ for
$i=1,2$.
Consider the pullback diagram \minCDarrowwidth.1cm
\eq{43c}{\begin{CD} 0 @>>> V_\C @>>> V\otimes \sO_\H @>d>> V \otimes \Omega^1_\H @>>> 0 \\
@. @| @AAA @AA 1\mapsto \delta(mot) A \\
0 @>>> V_\C @>>> N @>>> \C_\H @>>> 0
\end{CD}
}
The idea is we view the bottom sequence in~\eqref{43c} as underlying an extension of variations of Hodge structure over $\H$ and we compute the amplitude in the usual way~\eqref{31a} by lifting (ie. integrating) $\delta(mot) = (6\sqrt{3}/\pi)eis_\psi$ and pairing the resulting element in $V\otimes \sO_\H$ with $\omega$. We can write 
\eq{}{\omega = \varpi_r dz = \varpi_r (\tau\ve_1+\ve_2)\,.
}
Here $\varpi_r$ is the real period as in
equation~\eqref{e:realp}. Formally, we find that the amplitude is
given by the following Eichler integral
\ml{amplitude1}{\text{amplitude} = \\
{6\sqrt{3}\over\pi(2\pi i)^2}\varpi_r\left\langle \tau\ve_1+\ve_2, \int_\tau^{i\infty} \sum_{(a,b)\neq (0,0)} \frac{\psi(a)(x\ve_1+\ve_2) dx}{(ax+b)^3}\right\rangle = \\
{12i\sqrt{3}\over(2\pi i)^2}\varpi_r\,\int_\tau^{i\infty} \sum_{(a,b)\neq (0,0)} \frac{\psi(a)\,(\tau-x)}{(ax+b)^3} \,dx.
}
We do the integration term by term substituting
\eq{29d}{\int_\tau^{i\infty} \frac{dx}{(ax+b)^3} = \frac{-1}{2a(b+a\tau)^2};\quad \int_\tau^{i\infty} \frac{x dx}{(ax+b)^3} = -\frac{b+2a\tau}{2a^2(b+a\tau)^2}
}
Since $\psi(0)=0$, this yields
\eq{ampSum}{\text{amplitude} 
=({6i\sqrt{3}/(2\pi i)^2})\varpi_r\sum_{a\neq 0\atop b\in \Z} \frac{\psi(a)}{a^2(a\tau+b)}.
} 
Note there is a convergence issue here since the sum is not absolutely
convergent. We will treat this sum using the ``Eisenstein summation''
regularization following~\cite[eq.~(14) on page 13]{weil},  and we write
\eq{}{\lim_{N\to +\infty} \sum_{n=-N}^N \frac{1}{\tau+n} = \pi i\frac{q+1}{q-1};\quad q=\exp(2\pi i\tau).
}
Substituting in eq.~\eqref{ampSum}
\eq{8.32}{\text{amplitude} = ({-6\pi\sqrt{3}\over(2\pi i)^2})\varpi_r\sum_{a\neq 0}\frac{\psi(a)}{a^2}\frac{q^a+1}{q^a-1}
}
Since both $\psi(a)$ and $\frac{q^a+1}{q^a-1}$ are odd as functions of $a$, we can write this as
\eq{e:ampmot}{\text{amplitude} =\frac{12\pi\sqrt3}{(2\pi i)^2}\, \varpi_r\,
  \sum_{a\in\Z\atop a\neq0} {\psi(a)\over a^2} \,{1\over 1-q^a}
.
}
(Here of course we assume $|q|\le 1$. If $|q|>1$ we can write down a similar expression involving $q^{-1}$.)
Comparing with the expression for the sunset integral
in~\eqref{e:amplitude} we find the relation 

\begin{equation}\label{e:match}
  \textrm{amplitude}=  \cJs^2(t)+\textrm{periods}\,.  
\end{equation}
 %---------------------------------------------------------------------------
 \section*{Acknowledgements}
 P.V. would like to thank Claude Duhr and Lorenzo Magnea for
 encouraging him to analyze the sunset graph, and Don Zagier for having educated him about elliptic curves. He would like to thank  Herbert Gangl, Einan Gardi, and 
  Oliver Schnetz  for useful discussions.  S.B. would like to thank
  Sasha Beilinson, David Broadhurst, and Matt Kerr for their help. In
  particular, Beilinson's remark that the Eisenstein extension is
  determined by its residues played a central role in the
  argument. The technology of currents as developed by Kerr and
  collaborators in~\cite{kerrAJ}  provides an alternate approach to
  the motivic calculations.  
  
  The research of P.V. was supported by the ANR grant   reference QFT
  ANR 12 BS05 003   01, and the CNRS grant PICS number 6076.
 %---------------------------------------------------------------------------
 \appendix
%--------------------------------------------------------------------------
\section{Elliptic Dilogarithm}\label{sec:ellipticdilog}

In this appendix we recall the main properties of the elliptic
dilogarithms following~\cite{BlochCMR,ZagierElliptic,Zagier}.

Starting from the Bloch-Wigner dilogarithm 
\begin{eqnarray}
  D(z)&=& \Imm(\Li_2(z)+\log|z| \log(1-z))\cr
&=& {1\over2i}\left(\Li_2(z)-\Li_2(\bar z)+\frac12\, \log(z\bar z)\,
  \log\left(1-z\over 1-\bar z\right)\right)\,,  
\end{eqnarray}
this function is univalued real analytic on $\mathbb P^1(\mathbb
C)\backslash\{0,1,\infty\}$, and continuous on $\mathbb P^1(\mathbb
C)$.
This function satisfies the following relations
\begin{eqnarray}
  \label{e:Drel}
D(e^{i\theta})&=&Cl_2(\theta)= \sum_{n=1}^\infty {\sin(n\theta)\over n^2},\qquad
\theta\in \IR\\
D(z^2)&=&2\,\left(D(z)+D(-z)\right)\,.
\end{eqnarray}
We have the following six relations~\cite{Zagier}
\begin{eqnarray}
  \label{e:Dfunc}
  D(z)&=&-D(\bar z)=D(1-z^{-1})=D((1-z)^{-1})\cr
&=&-D(z^{-1})=-D(1-z)=-D(-z(1-z)^{-1})\,.
\end{eqnarray}
The $D(z)$ function satisfies
\begin{equation}
  dD(z)= \log|z| d\arg (1-z) - \log|1-z| d\arg(z)\,.
\end{equation}
%

%%%%%%%%%%%%%%%%%%%%%%%%%%%%%%%%%%%%%%%%%%%%%%%%%%%%%%%%%%%%%%%%%

\end{document}